\documentclass[11pt]{article}
\usepackage[letterpaper]{geometry}
\usepackage[parfill]{parskip}
\usepackage{amsmath,amsthm,amssymb,bbm}
\usepackage{mathtools}
\usepackage{cases}
\usepackage{graphicx}
\usepackage{tabularx}
\usepackage{microtype}
\usepackage{enumitem}

\usepackage{authblk}

\usepackage{url}
\usepackage[colorlinks,citecolor=blue,urlcolor=blue,linkcolor=blue,linktocpage=true]{hyperref}
\usepackage{cleveref}
\crefformat{equation}{(#2#1#3)}
\crefrangeformat{equation}{(#3#1#4) to~(#5#2#6)}
\crefname{equation}{}{}
\Crefname{equation}{}{}

\usepackage[authoryear]{natbib}
\usepackage{multirow}

\newtheoremstyle{mythmstyle}
  {8 pt} 
  {3 pt} 
  {} 
  {} 
  {\bfseries} 
  {.} 
  {.5em} 
  {} 

\theoremstyle{plain}

\makeatletter
\def\thm@space@setup{%
  \thm@preskip=6pt plus 1pt minus 1pt
  \thm@postskip=\thm@preskip 
}
\makeatother

\newtheorem{theorem}{Theorem}[section]
\newtheorem{lemma}[theorem]{Lemma}
\newtheorem{corollary}[theorem]{Corollary}
\newtheorem{proposition}[theorem]{Proposition}
\newtheorem{remark}{Remark}
\newtheorem{example}{Example}
\newtheorem*{example*}{Example}

\newtheorem*{definition*}{Definition}
\newtheorem*{remark*}{Remark}
\crefname{definition}{\textbf{definition}}{definitions}
\Crefname{definition}{Definition}{Definitions}
\crefname{assumption}{\textbf{assumption}}{assumptions}
\Crefname{assumption}{Assumption}{Assumptions}

\begin{document}
\allowdisplaybreaks
\title{Cross-Validation with Confidence}

 \author[1]{Jing Lei}
\affil[1]{Department of Statistics, Carnegie Mellon University}

\maketitle

\begin{abstract}
Cross-validation is one of the most popular model selection methods
in statistics and machine learning.  Despite its wide applicability,
traditional cross-validation methods tend to select overfitting models, 
due to the ignorance of the uncertainty in the testing sample.
We develop a new, statistically principled 
inference tool based on cross-validation that takes into account
the uncertainty in the testing sample.  This new method 
outputs a set of highly competitive candidate models containing the best
one with guaranteed probability. As a consequence, our method
can achieve consistent variable selection in a classical linear regression setting,
for which existing cross-validation methods require unconventional split ratios.
When used for regularizing tuning parameter selection, the method can provide
a further trade-off between prediction accuracy and model interpretability.
We demonstrate the performance of the proposed method in several simulated and 
real data examples.
\end{abstract}


\section{Introduction}
\label{sec:introduction}

Cross-validation \citep{Stone74,Allen74,Geisser75} is one of the most popular methods
for model and tuning parameter selection in statistics and machine learning, due to its
conceptual simplicity and wide applicability.  The basic idea of cross-validation
is to fit and evaluate each candidate model on separate data sets so that the 
performance evaluation is unbiased. Despite its wide use, it is well-known that the
traditional cross-validation methods, including the leave-one-out and V-fold variants,
tend to select models that over fit. Early theoretical studies of cross-validation
\citep{Shao93,Zhang93,Yang07} indicate that, under a low-dimensional linear model, cross-validation cannot consistently select the correct model unless 
the training-testing split ratio tends to zero, which precludes both the 
leave-one-out and V-fold cross-validation.  Although these theory provide deep and
important insight for cross-validation procedures, in practice it is very rare to
use a small split ratio. In fact, most statistical packages use default
split ratios such as nine-to-one (10-fold cross-validation) or four-to-one (5-fold 
cross-validation).  These conventional split ratios remain popular primarily because
a small split ratio greatly reduces the training sample size and usually results in
much less accurate model fitting. 

It has been observed in the literature that cross-validation overfits when it fails to take into account the uncertainty in the testing sample. In particular, overfitting occurs if a smaller average test error comes as a result of the random sampling of testing data instead of a truly superior fit.
To take into account the testing sample uncertainty and fix the issue of overfitting, we develop a hypothesis testing framework for cross-validation, which we call cross-validation with confidence (CVC).
For each candidate model $m$, CVC tests the null hypothesis that
the regression functions estimated from model $m$ 
have the smallest predictive risk among all fitted candidate models, and calculates a $p$-value 
by comparing the cross-validated residuals of \emph{all} candidate models. 
The subset of candidate models for which the null hypotheses are not rejected is a confidence set for model selection.

Depending on the context, cross-validation, and hence CVC, can be used for different purposes.  Here we focus on two most common contexts of cross-validation: model selection, and tuning parameter selection for risk minimization.

\paragraph{Model selection.}
Model selection is concerned with finding the true model, or the one closest to the truth, from a given set of candidate models. Here the set of candidate models is usually discrete and fixed beforehand.  For example, the candidate set may consist of all subsets of a given collection of covariates, or all polynomials less than a certain degree.
A classical and well studied problem in the cross-validation literature is consistent variable selection in linear regression.  Even in the low-dimensional case,
cross-validation with a conventional split ratio is known to be inconsistent 
\citep{Shao93,Zhang93,Yang07}.  We show that the smallest model in the
confidence set output by CVC can achieve consistent variable selection even with a
conventional split ratio.

\paragraph{Tuning parameter selection and risk minimization.}
In many learning problems, the algorithm is indexed by a regularization parameter.  Such a tuning parameter can either be continuous, such as the $\lambda$ in the Lasso, or discrete, such as the number of steps in forward stepwise selection. The tuning parameter selection problem is concerned with finding the tuning parameter value, from a given finite candidate set, that leads to the smallest predictive risk.
Under this context, cross-validation essentially finds the tuning parameter value whose fitted models have small validated predictive risk.  We show that (i) if all the tests in the CVC procedure are conducted at type I error 
level $\alpha$ for some $\alpha\in(0,1)$, then the confidence set contains the best 
fitted model with probability at least $1-\alpha$, and  (ii) with high probability, 
the confidence set only contains fitted models that are highly competitive.
The main challenge of the testing problem is the high correlation and vastly different scaling
between the cross-validated residuals.  Our test uses some recent results in high-dimensional 
Gaussian comparison developed by \cite{Cherno2,Cherno13many}, which allows us to provide
theoretical guarantees for our procedure under weak assumptions.

It is worth noting that risk minimization and model selection are two very different goals, especially when regularization parameter selection is concerned.  Regularization typically works by reducing the variance at the cost of adding a small amount of bias, and the added bias may lead to an incorrect model due to the discrete nature of the model selection task. When cross-validation is used to select the regularization parameter, we are often in the context of risk minimization.  Without further assumptions, such as that in linear regression the non-zero regression coefficients are sufficiently large and the covariates are weakly correlated, it is generally hard to establish model selection consistency for cross-validation based methods. This is not the focus of the current paper. See \cite{WangZ11,FanT13,HuiWF15} for some examples of information criteria based methods that combines regularization parameter tuning and variable selection. 


\paragraph{Other related work.}
The problem of finding a set of plausible candidate models has been studied by several authors. \cite{HansenLN11} consider this problem for low-dimensional linear regression problems using a hierarchical testing procedure based on the $F$-test.  \cite{FerrariY15} extend and simplify this idea to cover high-dimensional problems, provided that there exists a good variable screening method.  \cite{Jiang08} develop the ``fence'' method which finds a set of candidate models that is likely to contain the true model.  All these methods only apply to the model selection problem with a fixed candidate set, and the output confidence set of models always includes the saturated model, which may be overfitting.

 There are a few heuristic attempts to fix the overfitting
issue of cross-validation. For example, \cite{TibshiraniT09} propose a heuristic bias 
correction for V-fold cross-validation, where the average cross-validated risk
 is offset by the standard error of the $V$ individual subsample cross-validated risk estimates.  
\cite{LimY16} use a stability-based criterion, which is the ratio between the empirical
variance of the fitted vectors and the squared $\ell_2$-norm
of the average fitted vector. Here the ``empirical variance'' and ``average'' are taken over
repeated fitting using subsamples. Both methods still output a single candidate model.
\cite{LeiGRTW16} describe a sample splitting method for assessing
variable importance in building predictive regression models, which also uses a hypothesis
testing idea. The theoretical properties are not clear for these methods.  
\cite{WassermanR09} propose a cross-validation
based variable selection procedure which uses a three-way sample split and a least squares
fit after Lasso.  \cite{FengY13} study model selection consistency of a version of 
cross-validation specifically designed for penalized likelihood estimators, where, in
addition to assumptions on the signal pattern and the design matrix, the split ratio is
required to vanish. 

Another focus of the cross-validation literature is risk consistency. In an early work,
\cite{Li87} show that cross-validation can select a model with prediction risk close to
the optimal candidate model. Similar results in high dimensional regression with Lasso
have been established recently by \cite{HomrighausenD13,ChetverikovL16} for $V$-fold
cross-validation procedures.  \cite{ChatterjeeJ15} study the risk consistency of
cross-validated Lasso with a specially designed 2-fold cross-validation variant. The
additional value of our CVC method is that it outputs a subset of candidate models with
close-to-optimal performance, so that the user has more flexibility in choosing a
particular model. For example, one can choose a model that is more parsimonious
than the one given by standard cross-validation, or one can pick a model in the confidence
set that contains a particularly desirable variable.

\paragraph{Organization of the paper.} In \Cref{sec:prelim} we briefly review 
cross-validation and illustrate the cause of overfitting. In \Cref{sec:method} we
present the CVC method and its extensions.  In \Cref{sec:theory} we provide general
coverage and performance guarantees of the CVC confidence set, and prove variable
selection consistency of the CVC method in a classical linear 
regression setting.
\Cref{sec:data} presents simulated and real data examples. \Cref{sec:discussion}
concludes the paper with some implementation details and open problems.
\Cref{sec:proof} contains technical proofs of main results.

\section{Review of cross-validation}
\label{sec:prelim}
Suppose we have data $D=(X_i,Y_i)_{i=1}^n$ independently drawn from a common distribution 
$P$ on $\mathbb R^p \times \mathbb R$ satisfying
$$
Y_i = f(X_i)+\epsilon_i
$$
where $f:\mathbb R^p\mapsto \mathbb R$ is an unknown function and $\epsilon$ satisfies $\mathbb 
E(\epsilon\mid X)=0$.  
We are interested in finding an approximation of $f$ so that we can predict the values of $Y$ 
for future observations of $X$.
Let $\hat f$ be an estimate of $f$.  We evaluate the quality of $\hat f$ by the predictive risk
$$
Q(\hat f) = \mathbb E\left[\ell(\hat f(X),Y) \mid \hat f\right]\,,
$$
where $(X,Y)$ is a future random draw from $P$, and $\ell(\cdot,\cdot)$ is a loss function.
A typical example of $\ell(\cdot,\cdot)$ is the squared loss $\ell(\hat y,y)=(\hat y-y)^2$.  

\subsection{Validation by sample splitting.}
Here we describe cross-validation using a unified notation for both tuning parameter selection and model selection.
Suppose we are given a finite set of candidate models $\mathcal M=\{1,2,...,M\}$, where the meaning of each $m\in\mathcal M$ can be a model or a tuning parameter value, depending on the context.  For each $m\in\mathcal M$ 
one can estimate $f$ by $\hat f_m$ using an estimating procedure specified by $m$. 

The basic idea of cross-validation originates from estimating the predictive risk using an independent sample.  First we split the index set $\{1,...,n\}$ randomly into a training set $I_{\rm tr}$
and a testing set $I_{\rm te}=I_{\rm tr}^c$.  The data is correspondingly split into a training 
data $D_{\rm tr}=\{(X_i,Y_i):i\in I_{\rm tr}\}$ and a testing data
$D_{\rm te}=\{(X_i,Y_i):i\in I_{\rm te}\}$. For each candidate model $m$, we obtain an estimate 
$\hat f_m$ using $D_{\rm tr}$.  The quality of $\hat f_m$ is approximated using the empirical 
average loss on $D_{\rm te}$:
$$
\widehat Q_{\rm ss}(\hat f_m)=\frac{1}{n_{\rm te}}\sum_{i\in I_{\rm te}}\ell(\hat f_m(X_i),Y_i)\,,
$$
where $n_{\rm te}$ is the size of $D_{\rm te}$.  Finally, cross-validation chooses the model
that gives the smallest empirical risk on testing data:
\begin{equation}\label{eq:m_ss}
\hat m_{\rm ss}=\arg\min_{m} \widehat Q_{\rm ss}(\hat f_m)\,.
\end{equation}
Here the subscript ``ss'' stands for sample splitting.
 
Conditioning on the sample split and the training data $D_{\rm tr}$, the Law of Large Numbers 
implies that $\widehat Q_{\rm ss}(\hat f_m)$ is close to $Q(\hat f_m)$.  Therefore, cross-validation 
approximately minimizes the predictive risk over all candidate estimates
$\{\hat f_m: m\in\mathcal M\}$.

\paragraph{Sources of uncertainty in sample-splitting validation.}
From \eqref{eq:m_ss} we can see that there are two separate sources of uncertainty when using $\hat m_{\rm ss}$ as our estimate of the best $m$.
First, the estimates $\hat f_m$ are obtained from a finite sample $D_{\rm tr}$, which may not be representative for the overall performance of $\hat f_m$ when taking into account the randomness of $D_{\rm tr}$.  We call it the \emph{estimation uncertainty}.

The second source of uncertainty comes from using $\hat Q_{\rm ss}(\hat f_m)$ to approximate $Q(\hat f_m)$.  This is a typical approximation error of sample mean, and the randomness is in the testing data $D_{\rm te}$.  We call it the \emph{validation uncertainty}.

\subsection{Reducing the estimation uncertainty: V-fold cross-validation.}
V-fold cross validation extends the idea of sample splitting by repeatedly fitting each candidate model on different subsets of the data.  V-fold cross-validation first divides $\{1,...,n\}$ into $V$ equal-sized subsets $I_1,...,I_V$, and for each $m\in\mathcal M$ obtains $\hat f_{m}^{(v)}$ using data points in $I_{-v}=I_v^c$.
Then the quality of $m$ as a fitting procedure is approximated by averaging the
predictive risk of the estimates:
$$
Q_{\rm cv}(m)=\frac{1}{V}\sum_{v=1}^V Q(\hat f_m^{(v)})
$$
and the model that gives the best estimates is
\begin{equation}\label{eq:best_m_cv}
m^*_{\rm cv}=\arg\min_{m\in\mathcal M} Q_{\rm cv}(m)\,.  
\end{equation}
In practice we do not have $Q_{\rm cv}(m)$ and have to approximate it using sample-split validated predictive risk of $\hat f_m^{(v)}$ and average over $v=1,...,V$:
$$
\widehat Q_{\rm cv}(m)=\frac{1}{n}\sum_{v=1}^V \sum_{i\in I_v} 
\ell(\hat f_m^{(v)}(X_i),Y_i)\,,
$$
and the best candidate is estimated by
\begin{equation}\label{eq:m_cv}
\hat m_{\rm cv}=\arg\min_{m\in\mathcal M}\widehat Q_{\rm cv}(m)\,.
\end{equation}
A key difference between V-fold cross-validation and sample-split validation is that by averaging the approximate predictive risk of $V$ different estimates 
$\{\hat f_m^{(v)}:1\le v\le V\}$, the V-fold cross-validation provides a more comprehensive assessment by exploring the variability of $\hat f_m$ due to the randomness of fitting sample.  This can be more clearly seen if we re-write $\widehat Q_{\rm cv}$ as follows.
$$
\widehat Q_{\rm cv}(m)=\frac{1}{V}\sum_{v=1}^V\left[\left(\frac{n}{V}\right)^{-1} \sum_{i\in I_v} 
\ell(\hat f_m^{(v)}(X_i),Y_i)\right]\approx \frac{1}{V}\sum_{v=1}^V Q(\hat f_m^{(v)})\,,
$$
where each term $(n/V)^{-1} \sum_{i\in I_v} 
\ell(\hat f_m^{(v)}(X_i),Y_i):=\widehat Q_{\rm ss}(\hat f_m^{(v)})$ is the sample-split validated predictive risk  of $\hat f_m^{(v)}$, which is close to $Q(\hat f_m^{(v)})$ by law of large numbers.
While both $\widehat Q_{\rm ss}(\hat f_m)$ and $\widehat Q_{\rm cv}(m)$ are unbiased estimate of the overall predictive risk if we use model (or tuning parameter index) $m$ to estimate $f$, the latter often has much smaller variance. 

The extreme case of $V=n$ is known as the leave-one-out cross-validation.  In modern 
applications it is often expensive and unnecessary to re-fit the model $n$ times, so most 
statistical packages use 10-fold or 5-fold cross-validation as default choices.


%
%

\subsection{Validation uncertainty.}\label{sec:validation_uncertainty}
We start our discussion from a simple example.
\begin{example}\label{exa:overfit}
  Let $Y=\mu+\epsilon$ where $\epsilon\sim N(0,1)$.  The covariate $X$ is irrelevant here.
Let  $\mathcal M=\{0,1\}$ where $m=0$ corresponds to the model $\mu=0$ so $\hat f_0\equiv 
0$, and $m=1$ corresponds to the model $\mu\in \mathbb R$ so $\hat f_1(x)\equiv 
\hat\mu_{\rm tr}$ with $\hat\mu_{\rm tr}$ being the training sample mean of $Y$. We consider 
the squared error loss.  
Now assume $\mu=0$, then $\hat f_0$ is always a better estimate than $\hat f_1$.  But
\begin{align*}
  \widehat Q(\hat f_1)-\widehat Q(\hat f_0) = \hat\mu_{\rm tr}^2 -2 \hat\mu_{\rm tr}\bar 
  \epsilon_{\rm te}\,,
\end{align*}
where $\bar \epsilon_{\rm te}$
is the testing data sample mean of $\epsilon$.  Because $\sqrt{n_{\rm tr}}\hat \mu_{\rm tr}$ 
and $\sqrt{n_{\rm te}}\bar \epsilon_{\rm te}$ are independent standard normal random variables, 
there is positive probability that $\widehat Q(\hat f_1)-\widehat Q(\hat f_0)<0$.
This probability becomes larger 
when $n_{\rm te}$ is smaller.
\end{example}

In \Cref{exa:overfit}, cross-validation makes decision based on the sign of $\widehat Q(\hat 
f_1)-\widehat Q(\hat f_0)$, which contains a true signal part $\hat\mu_{\rm tr}^2$ and a noise
part $-2\hat\mu_{\rm tr}\bar \epsilon_{\rm te}$.
Overfitting happens if the noise overrides 
the signal.  It is clear that in this case it is the uncertainty in the validation step that causes the error, since the estimates $\hat f_0$ has no uncertainty and we almost surely have $Q(\hat f_0)< Q(\hat f_1)$.

Since the problem is not due to the estimation error, then intuitively we cannot expect V-fold cross-validation to fix this issue.  In fact it is not hard to show that the probability of overfitting is bounded away from $0$ for all values of $V$ if V-fold cross-validation is used in \Cref{exa:overfit}.

In the context of model selection, one way to account for validation uncertainty is to use a training-validation split ratio close to zero.  In the context of \Cref{exa:overfit}, the results in \cite{Shao93,Zhang93,Yang07} suggest that $\hat m_{\rm ss}$ is consistent if $n_{\rm tr}/n\rightarrow 0$.  The intuition is that $\bar \epsilon_{\rm te}=O_P(n_{\rm te}^{-1/2})$, much smaller than $\hat\mu_{\rm}$, which has the order $n_{\rm tr}^{-1/2}$.

Although choosing a vanishing training ratio can theoretically address the problem of validation uncertainty for model selection, it is practically undesirable to fit the models with a much smaller sample size.  Indeed, the required split ratio cannot be achieved by any conventional V-fold cross-validation.  More importantly, using a vanishing training sample ratio may become problematic in the context of tuning parameter selection and risk minimization.  A common practice of cross-validation in tuning parameter selection is to apply the fitting procedure on the entire data set with the chosen tuning parameter as a final output. In many important cases, such as bandwidth selection for nonparametric regression \citep{Tsybakov09}, and choosing $\lambda$ for the Lasso \cite{DalalyanHL17}, the optimal tuning parameter may depend on the sample size.  Therefore, if $n_{\rm te}/n\rightarrow 0$, then the optimal $\lambda$ for a training sample of size $n_{\rm te}$ is unlikely to be a good choice applied to the entire sample.

\section{Cross-validation with confidence}
\label{sec:method}

%
%
%
Now we describe cross-validation with confidence, a procedure that accounts for the validation uncertainty without sacrificing the training sample size.  Our discussion starts from the simpler sample-split validation and then extends to the V-fold cross-validation.


\subsection{Sample-split validation with hypothesis testing}
 Here
 we consider sample-split validation and focus on the conditional inference given
the candidate
estimates $\{\hat f_m:m\in\mathcal M\}$ obtained from training data $D_{\rm tr}$. Recall that we evaluate
a fitted model $\hat f$ by $Q(\hat f)=\mathbb E(\ell(\hat f(X),Y)\mid \hat f)$. For each $m\in\mathcal M$, we consider a hypothesis 
testing problem
\begin{equation}\label{eq:hypothesis}
  \begin{array}{l}
H_{0,m}:~ Q(\hat f_m)\le Q(\hat f_{m'})~~\forall~m'\neq m\\
\quad\text{vs}\\
H_{1,m}:~ Q(\hat f_m)> Q(\hat f_{m'})~~\text{for some}~~m'\neq m\,.    
  \end{array}
\end{equation}
The hypothesis $H_{0,m}$ means that the fitted model $\hat f_m$ has the best predictive risk
among all fitted models, and $H_{1,m}$ means that there exists another fitted model whose predictive risk is strictly less than $\hat f_m$.

Let $\alpha\in(0,1)$ be a pre-chosen level of type I error. If we have obtained a $p$-value, 
denoted by $\hat p_{{\rm ss},m}$,
for the testing problem \eqref{eq:hypothesis} for each
$m$, the CVC procedure outputs the \textit{confidence set}
\begin{equation}\label{eq:admi}
  \mathcal A_{\rm ss} = \{\hat f_m: m\in \mathcal M,~\hat p_{{\rm ss},m}>\alpha\},
\end{equation}
which contains all the fitted models for which $H_{0,m}$ is not rejected.

\subsubsection{Calculating the $p$-values.}\label{sec:p_val}
Fix an $m\in\mathcal M$, and define random vector  $\xi_{m}=(\xi_{m,j}:j\neq m)$ as
\begin{equation}\label{eq:xi}
  \xi_{m,j}=\ell(\hat f_m(X),Y) - \ell(\hat f_j(X),Y),
\end{equation}
where $(X,Y)$ is a fresh independent sample from the same joint distribution.
Let $\mu_{m,j}=\mathbb E(\xi_{m,j}\mid \hat f_m,\hat f_j)$. Then the hypotheses
in \eqref{eq:hypothesis} can be written equivalently as
\begin{equation}\label{eq:hypo_xi}
  H_{0,m}: \max_{j\neq m}\mu_{m,j}\le 0,~~\text{vs}~~H_{1,m}: \max_{j\neq m}\mu_{m,j}>0\,.
\end{equation}
\eqref{eq:hypo_xi} is a multivariate mean testing problem.  There are a few 
challenges. First,
the dimensionality can be high, as the number of candidate models or tuning parameter values 
can be on the order of hundreds. Second, the coordinates of $\xi_{m}$ are often highly 
correlated, since they are
calculated from prediction errors of similar models fitted from the same training data and 
evaluated on the same testing data.  Third,
the coordinates of $\xi_m$ usually have vastly different scaling, depending on the similarity 
between model $j$ and model $m$.  For example, suppose we are considering linear regression with squared error as the loss function, if $\hat f_m$ and $\hat f_j$ are both
$\sqrt{n}$-consistent then $\xi_{m,j}$ has a variance of order $O(n^{-1})$.  But if $\hat f_m$ is
$\sqrt{n}$-consistent and $\hat f_j$ is seriously underfitting (e.g., a linear model missing a relevant variable, or Lasso with a very large $\lambda$), then the variance of $\xi_{m,j}$ can be as large as a constant.

Considering all these challenges, we derive our test statistic and $p$-value using a recent
result of Gaussian comparison and bootstrap of high-dimensional sample means due to 
\cite{Cherno1}, which
provides an accurate approximation of the maximum of a high-dimensional sample mean
using the maximum of a corresponding Gaussian random vector.

The $p$-value $\hat  p_{{\rm ss},m}$ is calculated as follows, using a studentized Gaussian multiplier bootstrap approach.

\begin{enumerate}
  \item For $i\in I_{\rm te}$, let $\xi_m^{(i)}$ be the realization of $\xi_m$ on the testing data point $(X_i,Y_i)$:
  $$
  \xi_m^{(i)}=(\xi_{m,j}^{(i)}:j\neq m)\,~~\text{where}~~\xi_{m,j}^{(i)}=\ell(\hat f_m(X_i),Y_i)-\ell(\hat f_j(X_i),Y_i)\,.
  $$
  \item For each $j\neq m$, let $\hat\mu_{m,j}$ and $\hat \sigma_{m,j}$ be the sample
  mean and standard deviation of $\{\xi_{m,j}^{(i)}:i\in I_{\rm te}\}$.
  \item Let $T = \max_{j\neq m}\sqrt{n_{\rm te}} \frac{\hat\mu_{m,j}}{\hat\sigma_{m,j}}$.
  \item For $b=1,...,B$
  \begin{enumerate}
    \item Generate iid standard Gaussian random variables $\zeta_i$, $i\in I_{\rm te}$.
    \item Let 
    $$T^*_b=\max_{j\neq m}\frac{1}{\sqrt{n_{\rm te}}} \sum_{i\in I_{\rm te}} 
    \frac{\xi_{m,j}^{(i)}-\hat\mu_{m,j}}{\hat\sigma_{m,j}}\zeta_i .$$
  \end{enumerate}
   \item $\hat p_{{\rm ss},m}=B^{-1}\sum_{b=1}^B \mathbf 1(T^*_b > T)$.
\end{enumerate}
The Gaussian comparison and bootstrap results automatically take into account of the dimensionality and 
correlation
between coordinates of $\xi_{m}$.  The studentization addresses the issue of scaling difference.
A similar studentized Gaussian approximation has been considered  
  in an unpublished manuscript \cite{Cherno13many} for different applications.


\subsection{V-fold cross-validation with confidence.}
The V-fold version of the above hypothesis testing procedure is much more complicated, primarily due to the dependence between the fitted models $\hat f_{m}^{(v)}$ and the validated predictive loss $\ell(\hat f_{m}^{(v)}(X_i),Y_i)$ ($i\in I_{-v}$).

\subsubsection{A proxy group effect model.}
To further illustrate the issue, the conditional inference for sample-split validation is based on the fact that the validated predictive losses
$$
\left[\ell(\hat f_m(X_i),Y_i),~m=1,...,M\right],~~i\in I_{\rm te}
$$
are conditionally independent given $\hat f_1,...,\hat f_m$.
But in V-fold cross-validation we cannot condition on $\hat f_1^{(v)},...,\hat f_m^{(v)}$ as they are dependent with the data points in $I_{-v}$.  More precisely, for $1\le i < i'\le n$, let $i\in I_{v}$ and $i'\in I_{v'}$. There are two possibilities.
\begin{enumerate}
  \item [] Case 1. $v=v'$.  This is similar to the sample-split case, the dependence between $\ell(\hat f_{m}^{(v)}(X_i),Y_i)$ and $\ell(\hat f_{m}^{(v)}(X_i),Y_{i'})$ comes from $\hat f_{m}^{(-v)}$.  They are conditionally independent given $\hat f_{m}^{(v)}$. We denote this conditional distribution
  of $\ell(\hat f_{m}^{(v)}(X_i),Y_i)$ given $\hat f_m^{(v)}$ as $\hat{\mathcal F}_m^{(v)}$.
  \item [] Case 2. $v\neq v'$. Now $\ell(\hat f_m^{(v)}(X_i),Y_i)$ and 
  $\ell(\hat f_m^{(v')}(X_{i'}),Y_{i'})$ depend on each other in that $(X_i,Y_i)$ is used as a training sample point in the construction of $\hat f_m^{(v')}$, and $(X_{i'},Y_{i'})$ is used in $\hat f_m^{(v)}$. Moreover, $\hat f_{m}^{(v)}$ and $\hat f_m^{(v')}$ share the same $(1-2/V)n$ out of a total of $(1-1/V)n$ training sample points.
\end{enumerate}
In the second case described above, $\ell(\hat f_m^{(v')}(X_i),Y_i)$ and $\ell(\hat f_m^{(v)}(X_{i'}),Y_{i'})$ are not independent samples from $\hat{\mathcal F}_m^{(v')}$ and $\hat{\mathcal F}_m^{(v)}$ due to the dependence between $(X_i,Y_i)$ and $\hat f_{m}^{(v')}$ (as well as $(X_{i'},Y_{i'})$ and $\hat f_m^{(v)}$).  Intuitively, since $\hat f_m$ is a fitting procedure obtained from $(1-1/V)n$ sample points, it seems reasonable to believe that the dependence between $(X_i,Y_i)$ and $\hat f_m^{(v')}$ is weak, so that we can practically ignore such a dependence and treat $\ell(\hat f_m^{(v')}(X_i),Y_i)$ and $\ell(\hat f_m^{(v)}(X_{i'}),Y_{i'})$ as if they were independent samples from $\hat{\mathcal F}_m^{(v')}$ and $\hat {\mathcal F}_m^{(v)}$, respectively.

Following this intuition, for $1\le v\le V$, let $\hat{\mathcal F}^{(v)}$ denote the joint conditional distribution of
$$
\left[\ell(\hat f_m^{(v)}(X),Y)~,1\le m\le M \right]
$$
given $(\hat f_m^{(v)},~1\le m\le M)$.
The we proceed by treating the cross-validated predictive loss
$$
\ell_i:=\left[\ell(\hat f_m^{(v_i)}(X_i,Y_i)),~1\le m\le M\right]\sim \hat {\mathcal F}^{(v_i)}
$$
as independent samples from the distributions $\hat{\mathcal F}^{(v_i)}$, where $v_i$ is the fold id that the $i$th sample point belongs to. 

\subsubsection{Approximate $p$-value for V-fold CVC.}\label{sec:vfold_cvc_pval}
Now we have an (approximately) independent sample $\{\ell_i:1\le i\le n\}$ in $\mathbb R^M$ of size $n$, where within each group $I_v$ the subsample $\{\ell_i:i\in I_v\}$ forms an iid sample from $\hat{\mathcal F}^{(v)}$. Therefore, we can extend the $p$-value calculation to this case by taking out the group mean effects.

For $m\in\mathcal M$, the V-fold version of the hypothesis testing problem becomes:
\begin{align}
  &H_{0,m}: ~\frac{1}{V}\sum_{v=1}^V Q(\hat f_m^{(v)})\le \frac{1}{V}\sum_{v=1}^V Q(\hat f_{m'}^{(v)}),~~\forall~m'\neq m
  \label{eq:cvc_test}\\
  &\quad \text{vs}\nonumber\\
  &H_{1,m}: ~\frac{1}{V}\sum_{v=1}^V Q(\hat f_m^{(v)})> \frac{1}{V}\sum_{v=1}^V Q(\hat f_{m'}^{(v)})~~\text{for some}~m'\neq m\,.\nonumber
\end{align}

The $p$-values are calculated as follows.
\begin{enumerate}
  \item Define, for all $j\in \mathcal M\backslash\{m\}$, $$
\xi_{m,j}^{(i)}=\ell(\hat f_m^{(v_i)}(X_i),Y_i)-\ell(\hat f_j^{(v_i)}(X_i),Y_i),~~1\le i\le n\,.
$$
\item $\hat\mu_{m,j}^{(v)}=\frac{V}{n}\sum_{i\in I_v}\xi_{m,j}^{(i)}$, for all $j\neq m$, $1\le v\le V$. These are the estimated group mean effect.
\item $\tilde \xi_{m,j}^{(i)}=\xi_{m,j}^{(i)}-\hat\mu_{m,j}^{(v_i)}$, for all $j\neq m$ $1\le i\le n$. These are group-wise centered difference of cross-validated predictive loss.
  \item Let $\hat\mu_{m,j}=V^{-1}\sum_{v=1}^V \mu_{m,j}^{(v)}$ be the overall mean effect and $\hat \sigma_{m,j}$ be the sample
  standard deviation of $\{\tilde\xi_{m,j}^{(i)}:1\le i\le n\}$.
  \item Let $T = \max_{j\neq m}\sqrt{n} \frac{\hat\mu_{m,j}}{\hat\sigma_{m,j}}$.
  \item For $b=1,...,B$
  \begin{enumerate}
    \item Generate iid standard Gaussian random variables $\zeta_i$, $1\le i \le n$.
    \item Let 
    $$T^*_b=\max_{j\neq m}\frac{1}{\sqrt{n}} \sum_{v=1}^V\sum_{i\in I_{v}} 
    \frac{\xi_{m,j}^{(i)}-\hat\mu_{m,j}^{(v)}}{\hat\sigma_{m,j}}\zeta_i .$$
  \end{enumerate}
   \item $\hat p_{{\rm cv},m}=B^{-1}\sum_{b=1}^B \mathbf 1(T^*_b > T)$.
\end{enumerate}
The level $\alpha$ confidence set of candidate models corresponding to the testing problem \eqref{eq:cvc_test} is
\begin{align*}
  \mathcal A_{\rm cv}=\{m\in\mathcal M: \hat p_{{\rm cv},m}\ge\alpha\}.
\end{align*}

%

The way we compute $\hat p_{{\rm cv},m}$ guarantees that
the $\mathcal A_{\rm cv}$ is non-empty, and also reveals the relationship between standard cross-validation and
CVC.
\begin{proposition}\label{pro:non_empty}
  If $\alpha<0.5$, then
  \begin{align*}
    \lim_{B\rightarrow\infty}\mathbb P\left(\hat m_{\rm cv}\in\mathcal A_{\rm cv}\right)=1\,,
  \end{align*}
  where $\hat m_{\rm cv}$ is the standard cross-validation output given in \eqref{eq:m_cv}.
\end{proposition}
\begin{proof}
  Be definition of $\hat m_{\rm cv}$, we always have $T\le 0$.  Because $\alpha<0.5$, the upper $\alpha$ quantile
  of the maximum of a zero-mean Gaussian vector ($T_b^*$) must be positive.  The claim then follows
  from the Weak Law of Large Numbers.
\end{proof}
Following the same reasoning, we can see that among all $m\in\mathcal M$, the test statistic $T$ is only non-positive when $m$ is the cross-validation
choice $\hat m_{\rm cv}$.  Thus, empirically $\hat m_{\rm cv}$ almost always gives the largest $p$-value among all
candidate models.

%

\subsection{Implementation details}
\paragraph{Computational complexity of CVC.}  The V-fold CVC consists of two main parts. The first part is to split the sample, compute the estimates by holding out each fold, and compute the cross-validated predictive loss for each sample point.  This part is also required in ordinary V-fold cross-validation.  The additional computing required by CVC is in the second part, where the approximate $p$-values are calculated for each candidate model or tuning parameter value.  In this step, for each $m\in\mathcal M$, the algorithm will calculate the difference of predictive loss between $m$ and all the other $M-1$ candidates.  This requires $O(n(M-1))$ subtractions of two real numbers.  The re-centering and standardization costs another $O(n(M-1))$ operations.  The multiplier bootstrap with bootstrap sample size $B$ will require $O(n(M-1)B)$ operations.  Thus the total additional computing complexity of CVC is $O(nM^2B)$.  This additional cost is independent of the dimensionality of the problem.  In our numerical experiments, the total running time of CVC is usually a few (less than $5$) seconds when $n$ is a few hundreds, $M=50$, and $B=200$.  This is comparable to other model selection methods based on bootstrap and resampling methods such as the ``fence'' method \cite{Jiang08}.

\paragraph{Acceleration using inequality selection.}
The extra computing complexity of CVC has a quadratic dependence on $M$, the total number of candidate models or tuning parameter values.  This additional cost can be substantial when $M$ is large.  However, when $M$ is large, it is more likely a large proportion of candidate models or tuning parameters will lead to poor estimates.  Thus it is possible to eliminate such models quickly without using bootstrap by applying the inequality screening method developed in \cite{Cherno13many}.

More specifically, fix $\alpha'\in(0,1)$, and for each $m\in\mathcal M$ $j\neq m$,
define
$$
\hat{\mathcal J}_m=\left\{j\in\mathcal M
\backslash\{m\}:\sqrt{n}\frac{\hat\mu_{m,j}}{\hat\sigma_{m,j}}\ge
 -2\frac{\Phi^{-1}(1-\alpha'/(M-1))}{\sqrt{1-[\Phi^{-1}(1-\alpha'/(M-1))]^2/n}}\right\}\,,
$$
where $\Phi$ is the standard normal cumulative distribution function.

Intuitively, the candidates in $\hat{\mathcal J}_m^c$ are ``obviously'' inferior to $m$ and there is no need to invoke the bootstrap comparison for these models.  As a result, one can modify the last three steps of the CVC $p$-value calculation as follows.

\textbf{CVC $p$-value calculation with inequality selection}
\begin{enumerate}
  \item [1'-4'] The same as steps 1-4 as given in \Cref{sec:vfold_cvc_pval}.
  \item [5'] Let $T=\max_{j\in \hat{\mathcal J}_m} \sqrt{n}\frac{\hat\mu_{m,j}}{\hat\sigma_{m,j}}$.
  \item [6'] For $b=1,...,B$
  \begin{enumerate}
    \item Generate iid standard Gaussian random variables $\zeta_i$, $1\le i \le n$.
    \item Let 
    $$T^*_b=\max_{j\in\hat{\mathcal J}_m}\frac{1}{\sqrt{n}} \sum_{v=1}^V\sum_{i\in I_{v}} 
    \frac{\xi_{m,j}^{(i)}-\hat\mu_{m,j}^{(v)}}{\hat\sigma_{m,j}}\zeta_i .$$
  \end{enumerate}
   \item [7'] $\hat p_{{\rm cv},m}=B^{-1}\sum_{b=1}^B \mathbf 1(T^*_b > T)$.
\end{enumerate}
In practice, we set $\alpha'=\alpha/10$. In the rare case $\hat{\mathcal J}=\emptyset$, which corresponds to the case that $m$ is ``obviously'' better than all competitors, then we can set $\hat p_{{\rm cv},m}=1$.
Existing theoretical results (e.g. Theorem 4.5 of \cite{Cherno13many}) guarantee that, under the same conditions in \Cref{cor:coverage_vfold} below, the $p$-value calculated with inequality screening will increase the type I error by no more than $3\alpha'$, provided that $\alpha'$ and $\alpha$ are not too small.

\section{Theoretical Properties}
\label{sec:theory}
In our theoretical analysis, we first investigate the properties of the confidence sets $\mathcal A_{\rm ss}$ and $\mathcal A_{\rm cv}$ in terms of risk minimization. Then we consider model selection using the CVC method, under the context of subset selection for classical linear models.
%

We always assume that the test sample size is proportional to the total sample size: $n_{\rm te}/n \in [\delta,1-\delta]$
for some $\delta\in(0,1/2)$.  This is satisfied for each fold in the $V$-fold cross-validation with a constant $V$.

\subsection{Quality guarantees of confidence sets.}\label{sec:theory_conditional}
\subsubsection{Sample-split validation analysis}
In the analysis of sample-split validation, we condition on the fitted models $\{\hat f_m:m\in\mathcal M\}$.
Our main technical assumption is the tail behavior of the difference of prediction loss
$\xi_{m,j}=\ell(\hat f_m(X),Y)-\ell(\hat f_j(X),Y)$, where the randomness is in the future data pair $(X,Y)$.
We assume sub-exponential tail probabilities for standardized $\xi_{m,j}$, characterized by the Orlicz $\psi_1$ norm.
 Given $a\ge 1$,  the Orlicz $\psi_a$ norm of a  random variable $Z$ is
\begin{equation*}
  \|Z\|_{\psi_a} = \inf\left\{C>0:\mathbb E \exp\left(\frac{|Z|}{C}\right)^a\le 2\right\}\,.
\end{equation*}
Having a finite Orlicz $\psi_1$ norm is equivalent to having a sub-exponential tail, and implies finite moments of all orders and exponential concentration of sample mean. A finite
Orlicz $\psi_2$ norm is equivalent to a sub-Gaussian tail. See \cite{vdvWellner} and \cite{GeerL13} for further discussion on the Orlicz norm.

Following common practice, our analysis of the multiplier bootstrap assumes that $B$ is large enough
so the bootstrap sampling variability is negligible. That is, we take $\hat p_{{\rm ss},m}$ as the limiting value when $B\rightarrow\infty$.
\begin{theorem}\label{thm:conditional}
  Let 
  $\mu_{m,j}=\mathbb E(\xi_{m,j}\mid \hat f_m,\hat f_j)$, $\sigma_{m,j}^2={\rm Var}(\xi_{m,j}\mid
  \hat f_m,\hat f_j)$.  Assume 
  \begin{equation}\label{eq:tail_assumption}
    \sup_{m\neq j} \left\|\frac{\xi_{m,j}-\mu_{m,j}}{\sigma_{m,j}}\right\|_{\psi_1}\le A_n,
  \end{equation}
  for all 
  $m\neq j$ with 
  \begin{equation}\label{eq:rate_requirement}
    \frac{(A_n\vee 1)^6 \log^7(M\vee n)}{n^{1-c}}\rightarrow 0\,,
  \end{equation}
  for some $c\in(0,1)$. The following hold.
  \begin{enumerate}
    \item If 
    $\max_{j\neq m}\frac{\mu_{m,j}}{\sigma_{m,j}}\le r_n \sqrt{\frac{1}{n\log(M\vee n)}}$ 
    for some $r_n=o(1)$, then
    $\mathbb P(m\in\mathcal A_{\rm ss})\ge 1-\alpha+o(1)$.
    \item If $\alpha \in(n^{-1},1)$ and 
    $\max_{j\neq m}\frac{\mu_{m,j}}{\sigma_{m,j}}\ge c' (A_n\vee 1)\sqrt{\frac{\log (M\vee n)}{n}}$ 
    for a constant $c'$ depending only on $n/n_{\rm te}$, 
    then $\mathbb P(m\in\mathcal A_{\rm ss})=o(1)$.
  \end{enumerate}
\end{theorem}
We prove \Cref{thm:conditional} in \Cref{sec:proof_1}. 
Part (1) of \Cref{thm:conditional} guarantees the inclusion of competitive estimates. If $\hat f_m$ is the best 
estimate in $\{\hat f_m:m\in\mathcal M\}$, then $H_{0,m}$ is true and $m$ will be included in $\mathcal A_{\rm ss}$ with probability at least $1-\alpha+o(1)$ by just
taking $r_n=0$ in part (1).  In fact, 
the proof reveals that in the trivial case where $\mu_{m,j}=\mathbb E(\xi_{m,j}\mid \hat 
f_m,\hat f_j)=0$ for all $j\neq m$, we have $\mathbb P(m\in \mathcal A_{\rm 
ss})=1-\alpha+o(1)$.


\begin{remark}The sub-exponential tail condition \eqref{eq:tail_assumption} can be satisfied, for example, 
for linear regression with squared loss if the covariate $X$ and noise $\epsilon$ are sub-Gaussian. 
The Orlicz $\psi_1$ norm in \eqref{eq:tail_assumption} will depend on  
the estimated coefficients $\hat\beta_m$, which are typically restricted in a bounded set with high probability
on $D_{\rm tr}$.

In practice, the $\psi_1$ norm condition of \Cref{thm:conditional} is in general hard to verify exactly.  However, given the fact that 
  $\|Z\|_{q}\le q! \|Z\|_{\psi_1}$
  for any random variable $Z$, one can estimate $(q!)^{-1}\left\|(\xi_{m,m'}-\mu_{m,m'})/\sigma_{m,m'}\right\|_q$ for a few small integer values of $q$ (e.g., $q=1,2,3,4,...$) using the realized $\xi_{m,m'}$ obtained from the validation sample by plugging in $\mu_{m,m'}$ and $\sigma_{m,m'}$ with empirical estimates. If the estimates do not grow too fast when $q$ grows, then the $\psi_1$ norm condition seems plausible.
  In our simulations, the CVC method works reasonably well in linear regression with noise distribution being student's $t$ with three degrees of freedom.
\end{remark}

\subsubsection{Analysis of the group-effect proxy model}
Exact analysis of the V-fold cross-validation is hard due to the dependence between the $\xi_{m,j}^{(i)}$'s (using the notation in \Cref{sec:vfold_cvc_pval})
among different folds.  Empirically, we found that V-fold CVC
has much better power in eliminating suboptimal estimates. We can provide a partial justification by studying the $p$-values $\hat p_{{\rm cv},m}$ given in \Cref{sec:vfold_cvc_pval} under the group effect model where the dependence are ignored.

\paragraph{The group-effect model as a proxy.}
In particular, we assume the following model.
\begin{enumerate}
  \item [A1.] $(\xi_{m,j}^{(i)}:j\in\mathcal M\backslash\{m\})$ are independent across $i=1,...,n$, satisfying $\mathbb E(\xi_{m,j}^{(i)})=\mu_{m,j}^{(v)}$ and ${\rm Var}(\xi_{m,j}^{(i)})=(\sigma_{m,j}^{(v)})^2$ for all $i\in I_v$.
  \item [A2.] There exists a positive number $B_n$ such that $B_n^{-1}\le \frac{\sigma_{m,j}^{(v)}}{\sigma_{m,j}^{(v')}}\le B_n$ for all $m,j,v,v'$.
  \item [A3.] $\|(\xi_{m,j}^{(i)}-\mu_{m,j}^{(v)})/\sigma_{m,j}^{(v)}\|_{\psi_1}\le A_n$ for all $m,j$ and all $i\in I_v$.
\end{enumerate}

The following corollary is a straightforward extension of \Cref{thm:conditional}, providing partial justification of $\mathcal A_{\rm cv}$.
\begin{corollary}\label{cor:coverage_vfold}
Under Assumptions A1-A3, if 
\begin{equation}\label{eq:cvc_thm_condition}
\frac{(A_n\vee 1)^6 B_n^6 \log^7(M\vee n)}{n^{1-c}}\rightarrow 0  
\end{equation}
then the same results of \Cref{thm:conditional} hold for $\mathcal A_{\rm cv}$
with $\mu_{m,j}=V^{-1}\sum_{v=1}^V\mu_{m,j}^{(v)}$ and $\sigma_{m,j}^2=V^{-1}\sum_{v=1}^V(\sigma_{m,j}^{(v)})^2$.
\end{corollary}

The only difference from \Cref{thm:conditional} is the involvement of $B_n$ in \eqref{eq:cvc_thm_condition}. This is necessary because
$$
\left\|\frac{\xi_{m,j}^{(i)}-\mu_{m,j}^{(v)}}{\sigma_{m,j}}\right\|_{\psi_1}
=\left\|\frac{\xi_{m,j}^{(i)}-\mu_{m,j}^{(v)}}{\sigma_{m,j}^{(v)}}\right\|_{\psi_1}\frac{\sigma_{m,j}^{(v)}}{\sigma_{m,j}}\le A_n B_n\,.
$$


\subsection{Model selection consistency in classical linear models.}\label{sec:linear}
Now we show that CVC can be used to overcome the overfitting issue of standard cross-validation
in a classical linear regression setting.  Assume that the regression function is linear
\begin{equation}\label{eq:linear}
Y = X^T\beta +\epsilon\,,  
\end{equation}
where $X\in \mathbb R^p$ has covariance $\Sigma$, and $\epsilon$ is independent noise with mean 
zero and variance $\sigma^2$.
Here for brevity we assume that $p$ is fixed. The argument and results can be extended to the
case that $p$ grows as a small polynomial of $n$ using a union bound argument.

Given a collection of subsets $\mathcal J = \{J_1,...,J_M\}\subseteq 2^p$, we would like to find the $m^*$ such that $J_{m^*}=\{j:\beta_j\neq 0\}$, assuming that the true model is included in the candidate set. The standard cross-validation (using a single split) estimates $\hat\beta_m$ by applying a least squares fit using the training data and
covariates in $J_m$, and evaluates the model by the residual sum of squares on the testing data.
CVC works analogously with the same least squares fitting and uses the squared residual as loss 
function.
For any given $\alpha$, CVC outputs an confidence set of candidate models.  To select the correct model, we consider the most parsimonious  model in the confidence set
\begin{equation}\label{eq:m_cvc}
  \hat m_{\rm ssc}=\arg\min_{m\in\mathcal A_{\rm ss}}|J_{m}|\,,
\end{equation}
and
\begin{equation}\label{eq:m_cvc_vfold}
  \hat m_{\rm cvc}=\arg\min_{m\in\mathcal A_{\rm cv}}|J_m|\,.
\end{equation}
According to \Cref{pro:non_empty}, $\mathcal A_{\rm cv}$ (and also $\mathcal A_{\rm ss}$) is non-empty with high probability when $\alpha<0.5$. So $\hat m_{\rm cvc}$ (and $\hat  m_{\rm ssc}$) is always well-defined when $\alpha<0.5$.

We consider the following conditions.

\begin{enumerate}
  \item [B1.] $\lambda_{\min}>0$, where $\lambda_{\min}$ is the minimum eigenvalue of $\Sigma=\mathbb E(XX^T)$.  
  \item [B2.] $X$ and $\epsilon$ have finite sixth moments: $\max_{1\le j\le p}\mathbb E(X_{j}^6)<\infty$,
  $\mathbb E(\epsilon^6)<\infty$.
  \item [B3.] $\max_{1\le i\le n}H_{ii}=O_P(n^{-1/2})$, where $H=X(X^TX)^{-1}X$ is the $n\times n$ hat matrix.
\end{enumerate}
Condition B1 is necessary for the best model to be uniquely defined.
Condition B2 requires finite sixth moments for $X$ and $\epsilon$, which is slightly stronger than those in previous work of cross-validation \citep{Zhang93,Shao93}.  We need this in order to control the studentized mean effect in an non-asymptotic manner, which needs $\mathbb E(|\xi_{m,m'}|^3)<\infty$, while $\xi_{m,m'}$ is a quadratic function of $\epsilon$ and $X$ since we are using the squared loss function.  Similar versions of condition B3 have also appeared in previous works.  For example, Assumption D of \cite{Zhang93} requires $\max_{i}H_{ii}=o(1)$.  Our additional rate requirement in B3 is not too stringent. For sub-Gaussian distributions, the maximum diagonal entry of the hat matrix is upper bounded with high probability by $C\log n/n$ for some constant $C$.


\begin{theorem}\label{thm:linear}
  Assume model \eqref{eq:linear} and assumptions B1-B3 hold.  Let
  $\hat m_{\rm ssc}$ be given as in \eqref{eq:m_cvc} with $\mathcal A_{\rm ss}$ being the sample split validation confidence
  set using type I error level $\alpha_n\rightarrow 0$ and $\alpha_n > n^{-1}$. Then $\mathbb P(\hat m_{\rm ssc}=m^*)\rightarrow 1$.
\end{theorem}
We prove \Cref{thm:linear} in \Cref{sec:proof_2}. One challenge is to show that
$\hat f_{m^*}$ has the smallest risk with high probability, which involves the 
Kolmogorov-Rogozin anti-concentration
inequality \citep{Rogozin61} since we do not assume Gaussianity for $X$ or $\epsilon$.

\begin{corollary}\label{thm:linear_vfold}
  Under the same conditions in \Cref{thm:linear}. Let $\hat m_{\rm cvc}$ be given as in \eqref{eq:m_cvc_vfold} using type I error level $\alpha_n\rightarrow 0$ and $\alpha_n>n^{-1}$. Then $\mathbb P(\hat m_{\rm cvc}=m^*)\rightarrow 1$.
\end{corollary}

\begin{remark}
Unlike \Cref{cor:coverage_vfold}, which is proved under a proxy group effect model, \Cref{thm:linear_vfold} is established for the genuine CVC procedure without ignoring the dependence between the validated predictive loss across different folds. This is possible since we do not ask for exact type I error control but simply a vanishing type I error probability.
\end{remark}

\begin{remark}[Choice of $\alpha_n$ in practice]
Our theory requires $\alpha_n$ to vanish as $n$ grows in order to control the probability of not including the best model. But this is necessary only if all candidate models are equally good. In practice it is often the case that the best model is strictly better than other candidate models, so the true type I error is smaller than the nominal level.  In practice we recommend using CVC with $\alpha_n=0.05$, the traditional level of statistical significance.  All numerical examples in this paper are conducted with $\alpha=0.05$.  
\end{remark}

\section{Numerical experiments}
\label{sec:data}
We illustrate the performance of the CVC method using synthetic and real data sets.  All cross-validation methods, including CVC are implemented using 5 folds.  CVC is implemented using inequality screening with $(\alpha,\alpha')=(0.05,0.005)$ and $B=200$. 

\subsection{Simulation 1: Subset selection consistency in linear models.}
In this simulation we demonstrate the model selection consistency of CVC to support the theory developed in 
\Cref{sec:linear}.  We adopt the synthetic data set in \cite{Shao93}, where the covariate $X$
has five coordinates including intercept, with a sample size $n=40$ (see Table 1 of \cite{Shao93} for the 
complete data set). The response variable $Y$ is generated by the linear model \eqref{eq:linear}
 with 
a pre-determined $\beta$ and independent noise $\epsilon$. We consider two true models
 $\beta^T=(2,0,0,4,0)$ and $(2,9,0,4,8)$ that correspond to a sparse model and a less sparse model.
The set of candidate models consists of all $16$ possible models that include the intercept term.  Two noise distributions are experimented: standard normal and student's $t$-distribution with three degrees of freedom.
For each type of noise distribution, we consider two noise-to-signal ratios.
The first one uses standard Gaussian and student's $t(3)$ to generate the noises.
The second one amplifies the noises by a factor of $2$.

The sample size $n=40$ is far from the regime of model selection consistency. 
To illustrate the effect of sample size, we generate additional sample points with $X$ drawn from a multivariate Gaussian 
distribution whose parameters are given by the sample mean and covariance of the original data set. We consider 
varying sample sizes between $n=40$ and $n=640$, where $n=40$ corresponds to the original data set.
We compare the performance of five methods, the ordinary cross-validation (cv), the most parsimonious model in $\mathcal A_{\rm cv}$ (cvc), BIC, the confidence set of models using $F$-test proposed by \cite{FerrariY15}, and the ``fence'' method \cite{Jiang08}.


\Cref{fig:shao_gaussian,fig:shao_t} summarize the frequency each method correctly selects the model, and the number of models that are contained in the confidence set, over $100$ 
independently generated data set.
For a very small sample size $n=40$, the CVC method does not perform as 
well as the other methods, as it often selects the underfitting models.  
The performance of 
CVC increases rapidly as the sample size $n$ increases, with
a perfect rate of correct selection as soon as $n$ reaches $320$.  Meanwhile, the performance of the 
standard cross-validation
does not improve as $n$ increases, with rates of correct selection staying away from $1$.  This agrees with the 
theory that CVC can consistently select 
the true model, while the standard cross-validation tends to overfit as long as the training sample ratio stays constant.

Overall, CVC is highly competitive when $n$ is moderately large and the true model is sparse.  Moreover, except in the high noise case with student's $t(3)$ noise, the CVC confidence set tends to be smaller than the F-test based confidence set \citep{FerrariY15} when $n$ is moderately large.

\begin{figure}
  \begin{center}
    \includegraphics[scale=0.7]{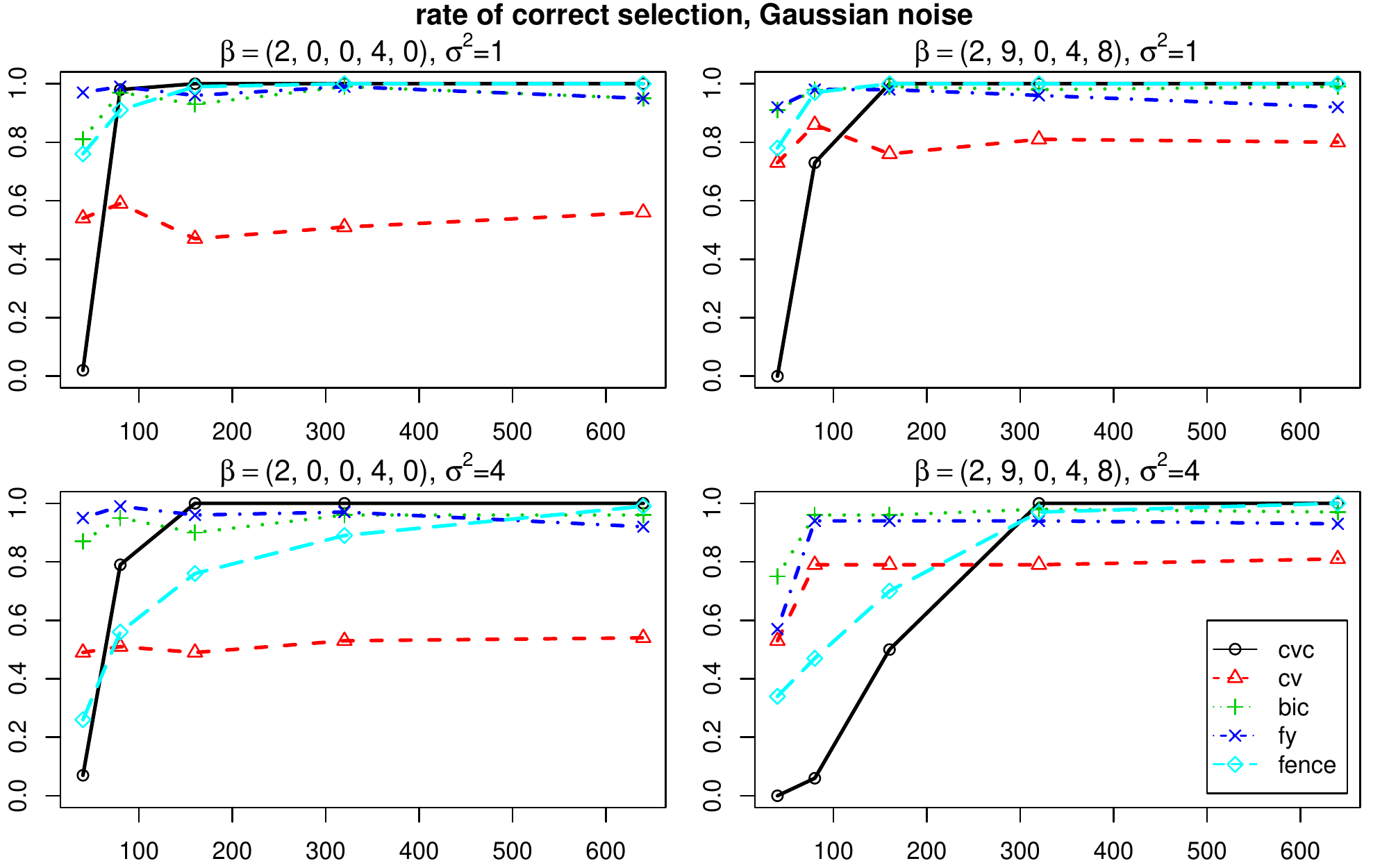}\\

\vspace{5pt}

    \includegraphics[scale=0.7]{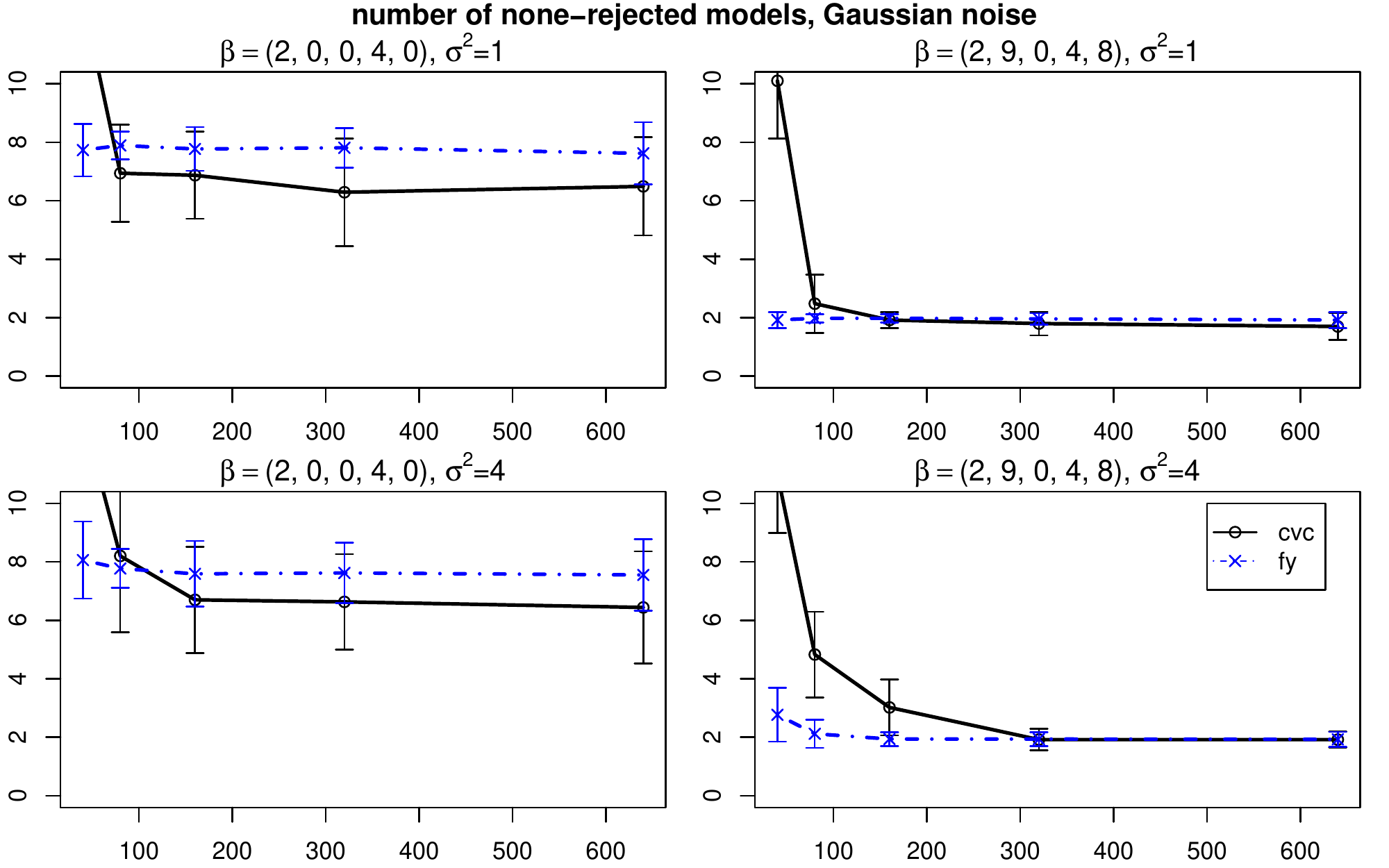}
    \caption{Simulation 1: linear regression with $p=5$ and Gaussian noise. First two rows: rate of correct subset selection as a function of sample size. Bottom two rows: number of none-rejected candidate models as a function of sample size.  
    \label{fig:shao_gaussian}}
  \end{center}
\end{figure}

\begin{figure}
  \begin{center}
    \includegraphics[scale=0.7]{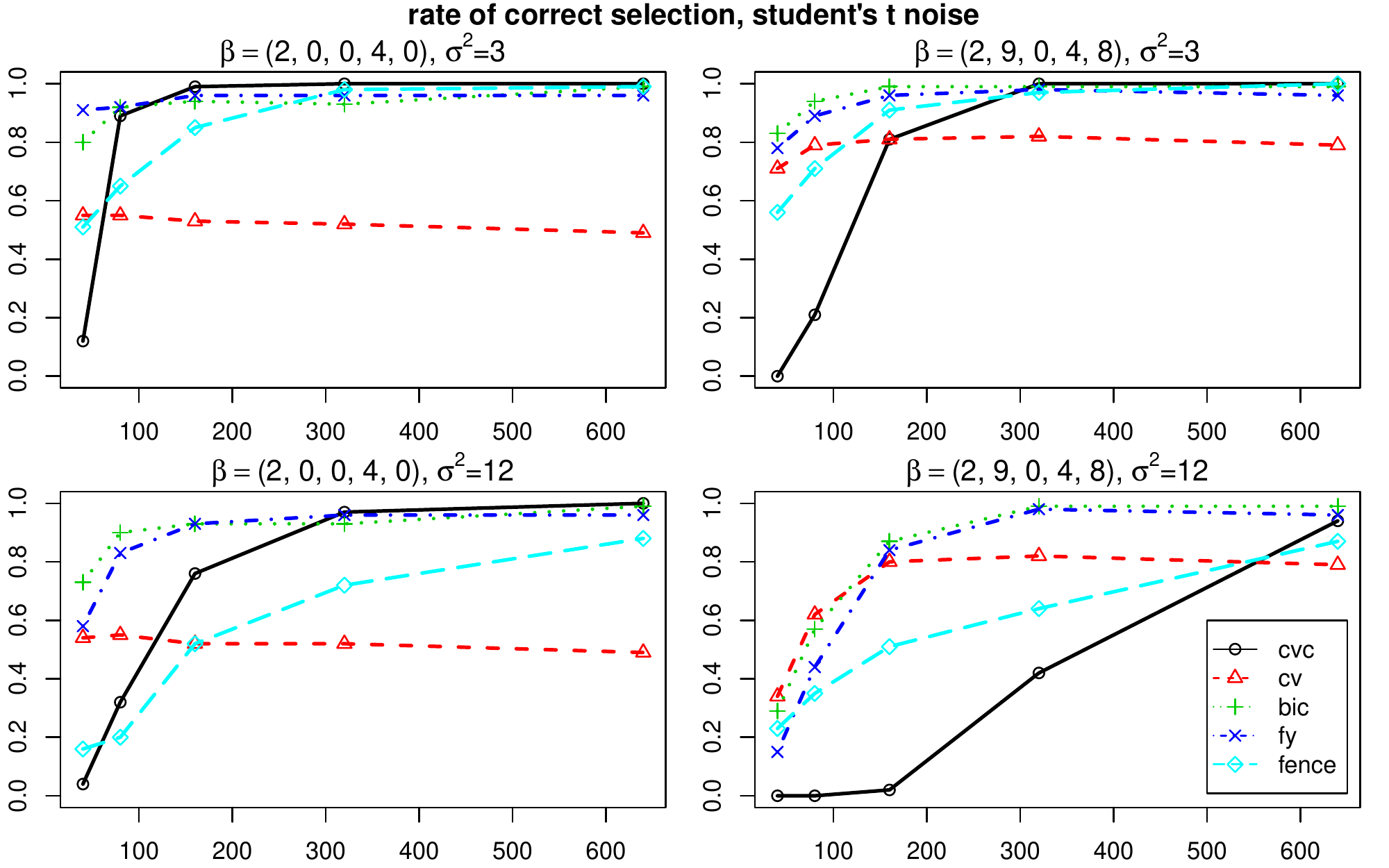}\\

\vspace{5pt}

    \includegraphics[scale=0.7]{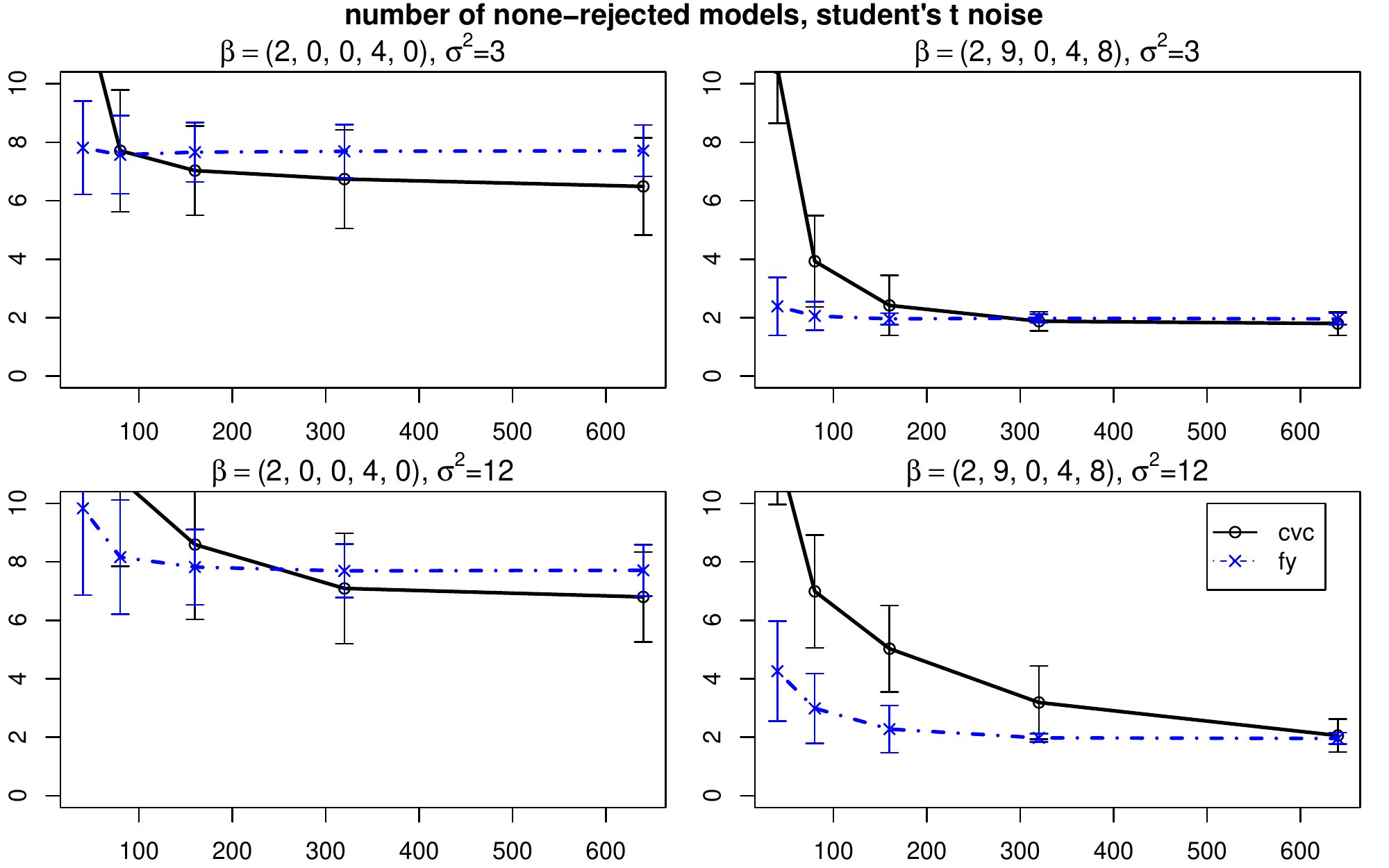}
    \caption{Simulation 1: linear regression with $p=5$ and student's $t_3$ noise. First two rows: rate of correct subset selection as a function of sample size. Bottom two rows: number of none-rejected candidate models as a function of sample size.
    \label{fig:shao_t}}
  \end{center}
\end{figure}

\subsection{Simulation 2: Tuning the Lasso for risk minimization}\label{sec:sim_lasso}
We use this simulation to demonstrate
the performance of CVC in choosing the Lasso tuning parameter for predictive risk minimization.  
We generate data from a linear regression model \eqref{eq:linear},
where $X\in\mathbb R^{200}$ has a multivariate Gaussian distribution $N(0,
\Sigma)$, and $\epsilon$ is an independent standard Gaussian noise.
We consider the squared loss $\ell(\hat y,y)=(\hat y-y)^2$.
 For any estimate $\hat \beta$ of $\beta$, 
the linear model
leads to a closed-form formula for $Q(\hat\beta)$:
$$
Q(\hat\beta) = \mathbb E \left[(Y-X^T\hat\beta)^2\mid\hat\beta\right] = (\hat\beta-\beta)^T\Sigma(\hat\beta-\beta)+1\,.
$$
We consider two settings of $\Sigma$.  In the \emph{identity setting}, $\Sigma$ is the $200\times 200$
identity matrix.  In the \emph{correlated setting}, the diagonal entries of $\Sigma$ are $1$ and off-diagonal entries
are $0.5$.  We also consider two settings of $\beta$.  In each setting, the first $s$ coordinates of $\beta$ take value $\pm 1$ with randomly chosen signs, and the next $s$ coordinates are generated from $N(0,I_s)$, and the remaining coordinates are $0$.  In the \emph{sparse setting}, we set $s=5$, while in the \emph{dense setting}, we set $s=25$ so that a quarter of the coordinates of $\beta$ are non-zero.

For each generated data set we first obtain a sequence of $50$ tuning parameters using \texttt{R} package \texttt{glmnet} on the whole data set. Using these $50$ candidate values of $\lambda$, we compare four methods: the original $V$-fold cross-validation (``cv''), the 1-standard-error rule \citep[``1se'',][]{TibshiraniT09}, the estimation stability method \citep[``es'',][]{LimY16}, and the most parsimonious model in the CVC confidence set with $\alpha=0.05$ (``cvc''), defined as
$$
\hat\lambda_{\rm cvc}=\max\{\lambda:\lambda\in \mathcal A_{\rm cv}\}\,.
$$
We run the simulation with sample size $n=200$, and repeat $100$ times for
 each combination of $\beta$ and $\Sigma$.

\paragraph{Evaluation and comparison of cross-validation based methods.} In the context of risk minimization using cross-validation, a common practice is to first find a promising value of tuning parameter (such as the one given by cross-validation or its variants) and then obtain a final estimate using the entire data set and the selected tuning parameter value.  This approach has been used in some existing modifications of cross-validation, including the 1-standard-error rule \citep{TibshiraniT09} and the stability rule \citep{LimY16}.  In the case of CVC and Lasso, we argue that one must take into account the sample size difference between the training sample used in cross-validation and the entire sample.  More precisely, in V-fold cross-validation, each candidate value of tuning parameter is validated based on the estimates obtained from a training sample of size $n(1-V^{-1})$, but the final estimate is obtained on a training sample of size $n$.  Intuitively, the optimal value of tuning parameter for one sample size is not necessarily also optimal for a different sample size, because when the sample size gets larger, less regularization is required. In fact, it has been suggested in the Lasso literature that the optimal value of $\lambda$ for both prediction and estimation is inversely proportional to the square root of the training sample size \citep{BickelRT09,vdGeer08,MeinshausenY09,DalalyanHL17}.  Therefore, in our implementation of CVC we make use of this observation by obtaining the final estimate on the entire data set  with tuning parameter value $\lambda=\sqrt{1-V^{-1}}\hat\lambda_{\rm cvc}$.

\begin{figure}
  \begin{center} \includegraphics[scale=0.36]{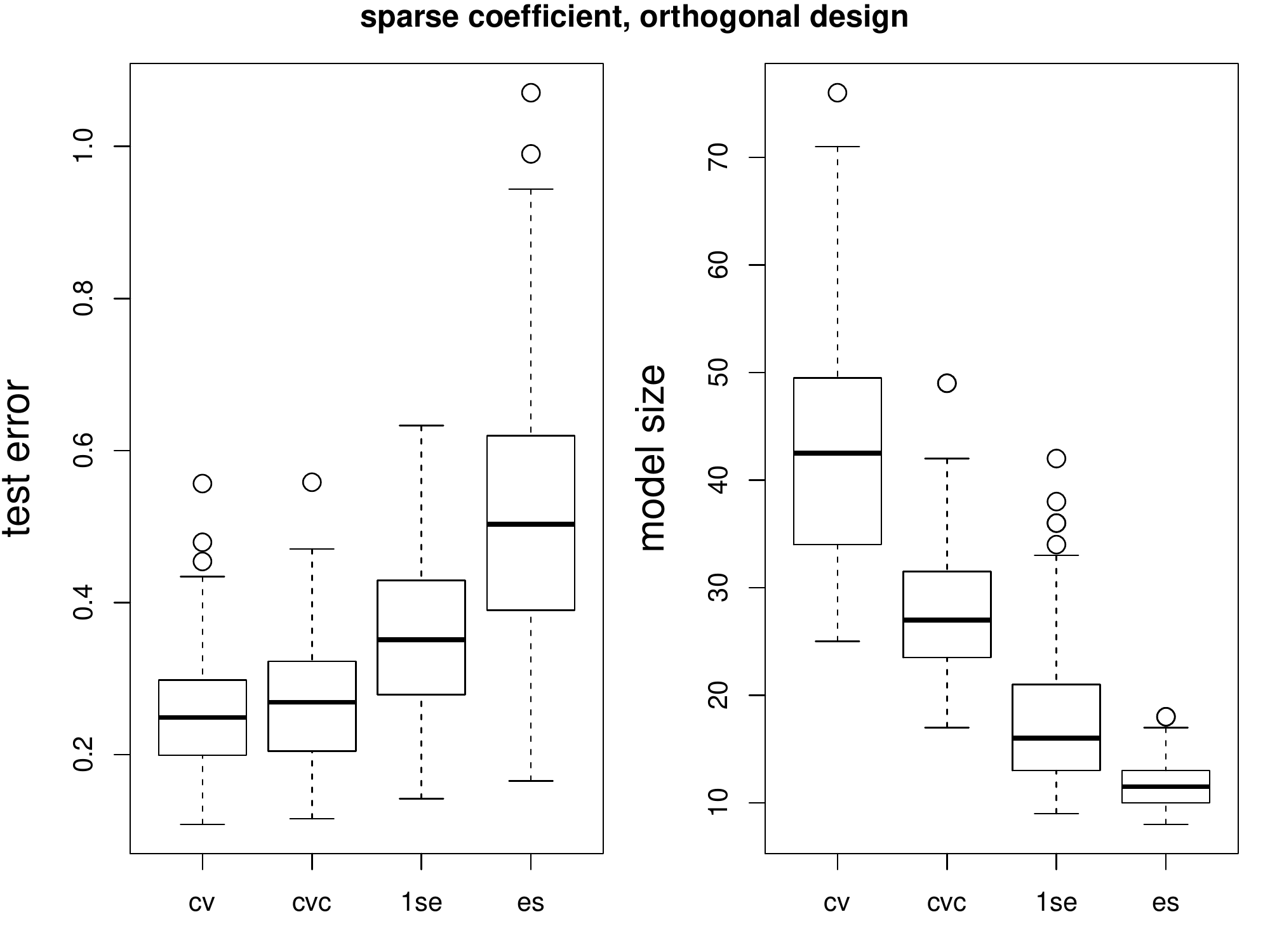}\includegraphics[scale=0.36]{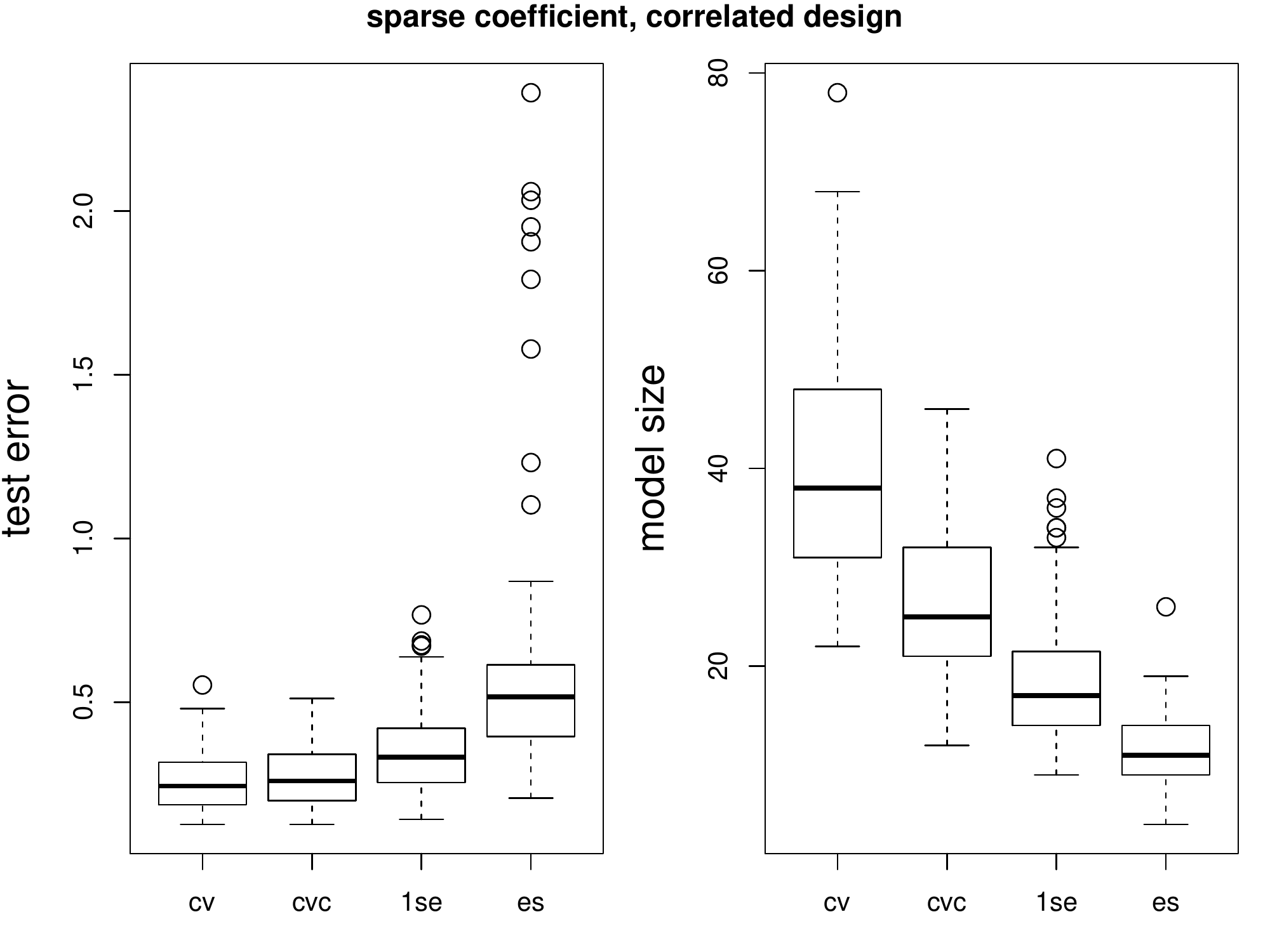}\\
                 \includegraphics[scale=0.36]{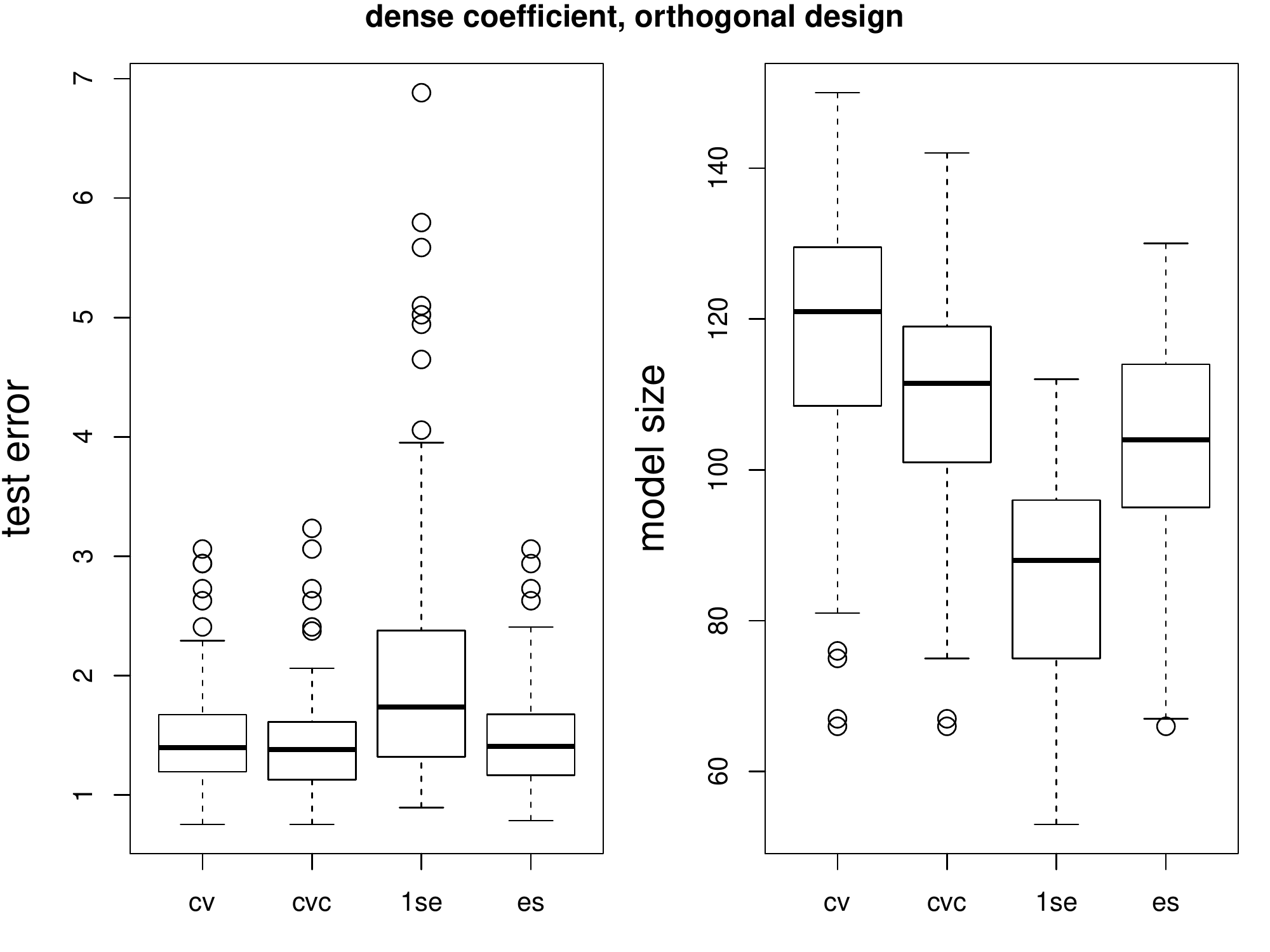}\includegraphics[scale=0.36]{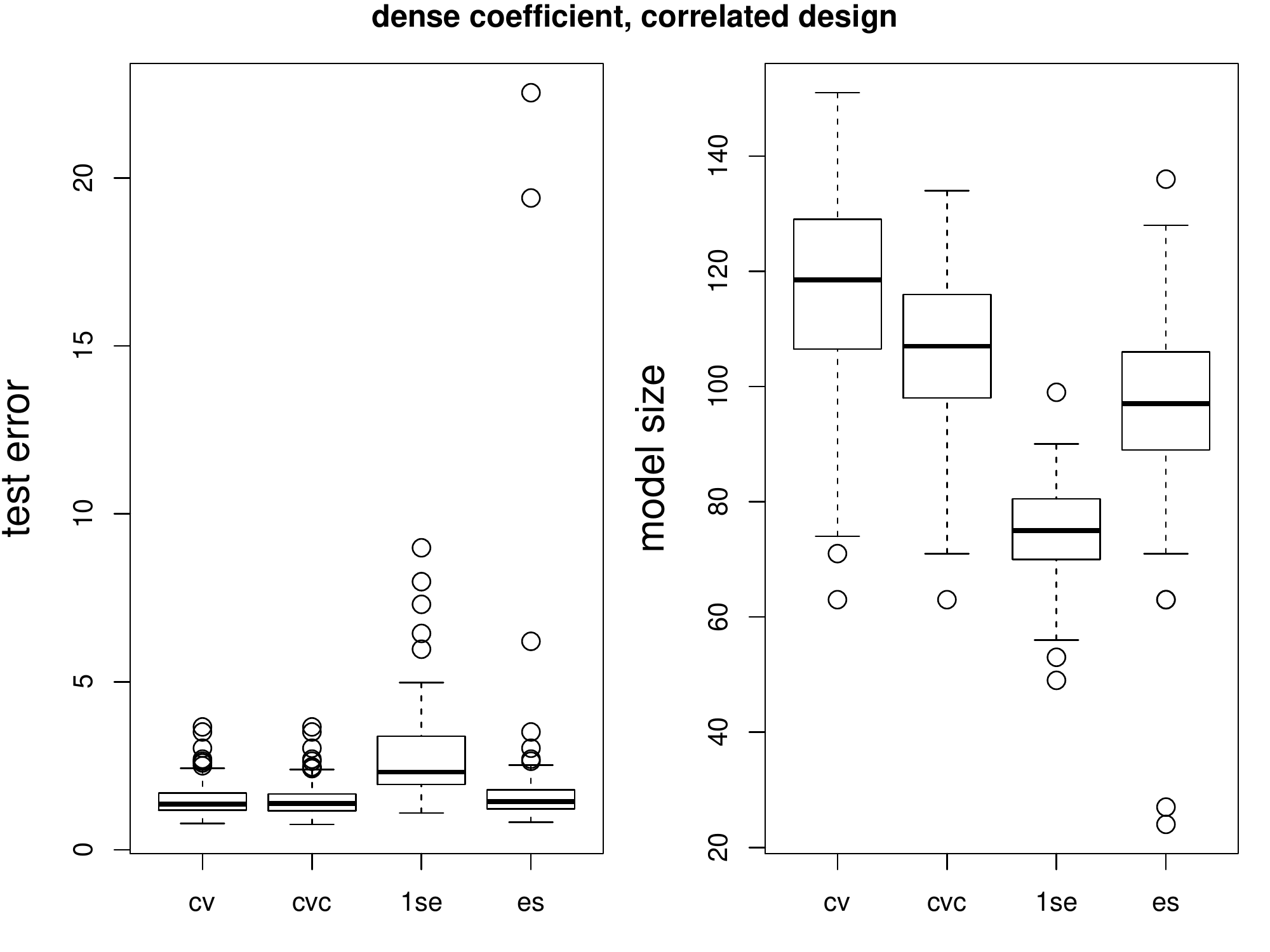}\\
  \caption{Risk and model size of four methods of Lasso tuning parameter selection.
  Top row: sparse model. Bottom row: dense model.  Columns 1-2: orthogonal design. Columns 3-4: correlated design.\label{fig:lasso_sim}}
  \end{center}
\end{figure}

\Cref{fig:lasso_sim} summarizes the performance.  We can see that the CVC method offers highly competitive prediction accuracy comparing to standard cross-validation, while using fewer predictors. Both CVC and ES have good adaptivity to the model sparsity, using more predictors when the model is dense.  Overall, CVC provides a better trade-off between interpretability and predictive accuracy compared to other methods. Moreover, the median value of $\mathcal A_{\rm cv}$ is between $4$ and $5$ for all settings, and the $\mathcal A_{\rm cv}$ covers the best tuning parameter value with a frequency almost exactly $0.95$ in all settings.

\subsection{The diabetes data example.}
We apply the CVC method to the diabetes data, which has been used in \cite{LARS} to illustrate the LARS algorithm in sparse linear regression.
The data contains $n=442$ sample points, with ten covariates including age, sex, body mass index, blood pressure, and six blood serum measurements.
If we consider all the quadratic terms as in the literature, there will be $64$ predictors (the sex variable is coded as $\{0,1\}$ 
and has no square term).
The response variable is a quantitative measure of diabetes progression one year after the covariates are measured.
Following common practice, we center and standardize all variables.

\Cref{fig:diabetes1} shows the result of 5-fold CVC with $\alpha=0.05$ on the full data set.
The base estimator is the Lasso, with $50$ equally spaced values (on the log scale) of the penalty parameter $\lambda$.
We plot the cross-validated test errors for all candidate values of $\lambda$.  The triangle points correspond to
the tuning parameter values in $\mathcal A_{\rm cv}$, and the solid triangle corresponds to the $\lambda$ chosen by standard
5-fold cross-validation.  5-fold CVC outputs five candidate values of $\lambda$, one of which gives a very close test error as the standard cross-validation but a more parsimonious
model fit.

\begin{figure}[t]
  \begin{center}
    \includegraphics[scale=0.5]{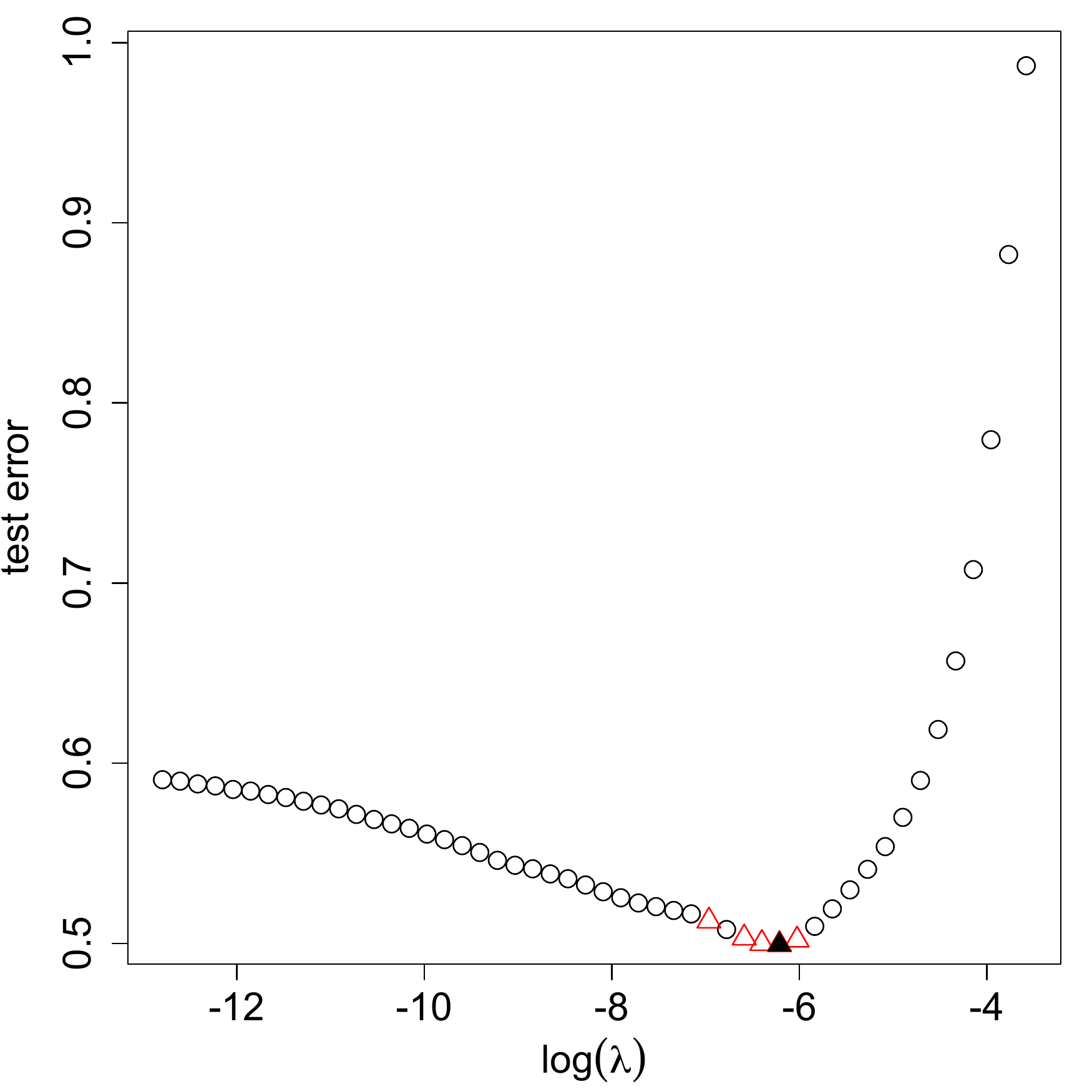}
    \caption{Diabetes data example.  The cross-validation test errors for 50 values of $\lambda$. The triangles correspond to
    the $\lambda$ values in $\mathcal A_{\rm cv}$ given by 5-fold CVC with $\alpha=0.05$.  The solid triangle corresponds
    to the $\lambda$ chosen by standard 5-fold cross-validation.  \label{fig:diabetes1}}
  \end{center}
\end{figure}

To investigate the predictive performance of the final estimate, we split the data so that $300$ sample points are used for estimation, including
implementing cross-validation or CVC, and the remaining $142$ sample points are held out for test error calculation.
To be specific, all competing methods (cv, cvc, 1se, es) are based on  5-fold cross-validation in picking $\lambda$, and then obtain the final 
estimate of regression coefficient
using Lasso with the chosen $\lambda$ on the full training sample of $300$ data points. As we explained in the previous subsection, the CVC method uses a scaled value of $\lambda$ to offset the training sample size difference when obtaining the final estimate. Then we use
the untouched set of $142$ hold-out sample points to calculate the test error for the final estimate. We repeat this experiment
$100$ times, each time with an independent hold-out split. 

The results are summarized in \Cref{fig:diabetes2}.
Comparing to the simulation results in \Cref{sec:sim_lasso}, the diabetes data set is more similar to a low-dimensional model, where more regularization can be used for better interpretation without losing prediction accuracy.
We see that the most parsimonious CVC model gives very close test error as the standard cross-validation, but has much fewer predictors.  The 1se method is also very competitive, with even more interpretable estimates and a slightly larger predictive risk.

\begin{figure}[t]
  \begin{center}
    \includegraphics[scale=0.7]{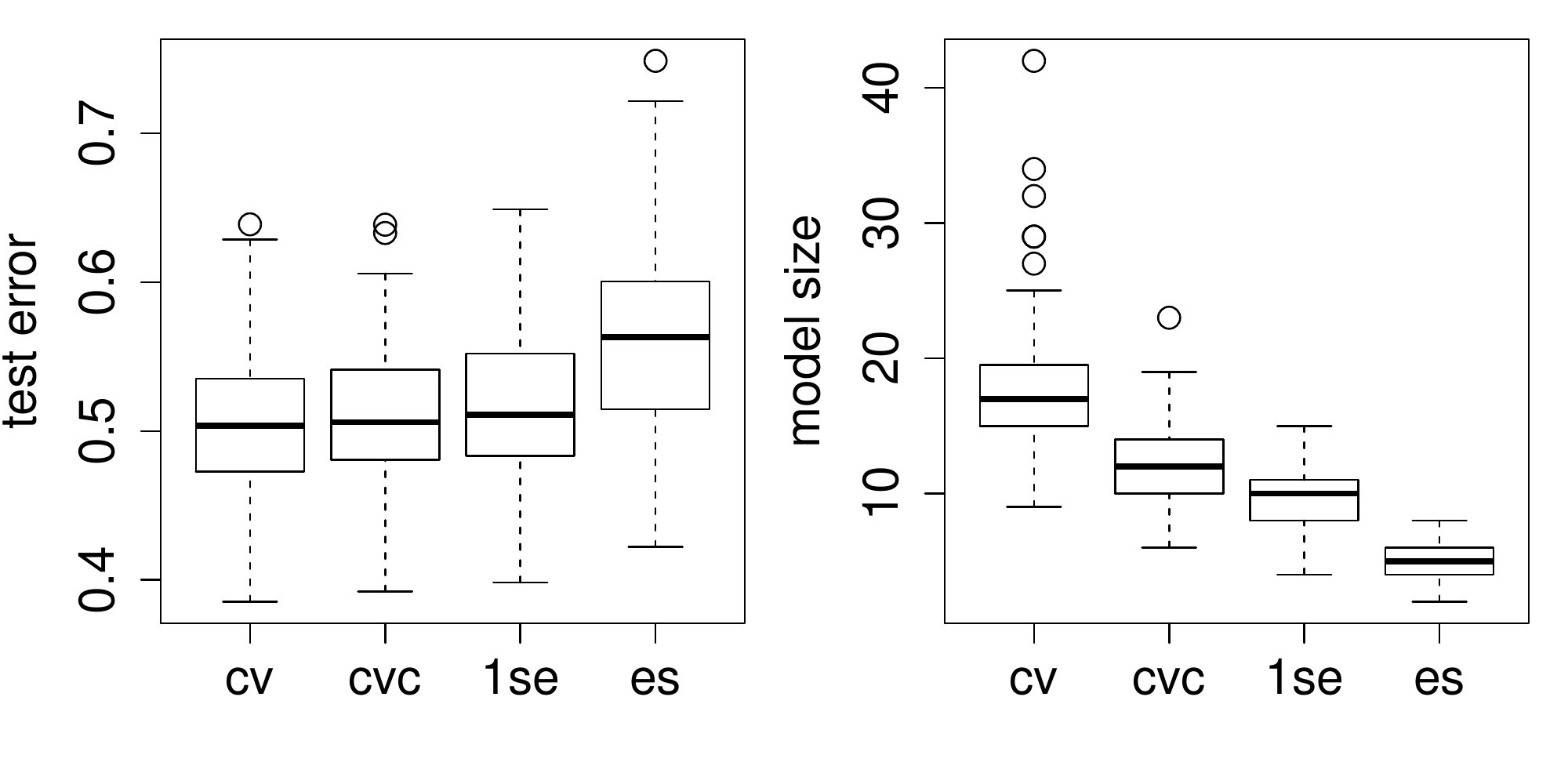}
    \caption{Diabetes data example.  Box-plots of hold-out test errors and 
    number of selected variables over 100 repeated hold-out splitting.
    Compared estimators are the standard cross-validation (``cv''), the most parsimonious model in $\mathcal A_{\rm cv}$ (``cvc''), the 1-standard-error rule (``1se''), and the estimation stability rule (``es'').
    \label{fig:diabetes2}}
  \end{center}
\end{figure}


\section{Discussion}
\label{sec:discussion}

The wide applicability of cross-validation makes it worthwhile to further extend our understanding of CVC under other
contexts. Several extensions shall be pursued in future work.
First, it would be interesting to study and extend CVC for high dimensional regression model selection problems. 
Second, in our consistency analysis of sparse linear regression in \Cref{sec:linear}, we considered independent 
noise.  If the response variable is
binary, then we need to extend the theory to cover heteroskedastic noises. 
Moreover, it is also possible to extend
the framework of CVC to unsupervised learning problems.  For example, in $K$-means or model-based clustering, one needs to specify the number of clusters.
Given a set of estimated clusters, the value of loss function at a test sample point can be set as the squared distance to the closest cluster center.  Similar extensions
can also be made in low rank matrix approximation using singular value decompositions \citep{OwenP09}
and network community detection \citep{ChenL16}, where various forms of cross-validation can be used for model selection.

\appendix
\section{Proofs}\label{sec:proof}
\subsection{Proof of \Cref{sec:theory_conditional}}\label{sec:proof_1}
In the construction of our $p$-value $\hat p_m$, the main test statistic
is the maximum of the studentized sample mean of $(\xi_{m}^{(i)}:i\in I_{\rm te})$.
The reference distribution is the maximum of a Gaussian vector with covariance
$\hat \Gamma_m$, the sample correlation matrix of $(\xi_{m}^{(i)}:i\in I_{\rm te})$.
Therefore, the main task is to prove that the maximum of a studentized sample mean
has similar distribution as the maximum of a Gaussian random vector with same covariance structure.
A result of this type is given as part of the proof of Theorem 4.3 of \cite{Cherno13many}.
Here we provide a self-contained proof using more compact notation and  under the assumption of
bounded Orlicz $\psi_2$ norm for $(\xi_{m,j}-\mu_{m,j})/\sigma_{m,j}$.  Finally, the unnumbered constants
$c,c'$ may vary from line to line.  Numbered constants, $c_1$, $c_2$, $c_3$, $C_1$, $C_2$, etc, correspond
to those specified in corresponding theorems.

\paragraph{Additional notation.}
Here we fix $m$ and will drop $m$ in the notation. Let $\mu$, $\Sigma$, $\Gamma$ be the mean, covariance and correlation matrix of the $M-1$ dimensional random vector $\xi=(\xi_{m,j}:j\neq m)$ (conditional on $D_{\rm tr}$). The corresponding sample versions are $\hat\mu$, $\hat\Sigma$, and $\hat\Gamma$.  Let $\sigma_j^2$ be the $j$th diagonal element of $\Sigma$, and $\hat\sigma_j$ the
sample version. 
Finally, let $N=M\vee n$.
\begin{lemma}\label{lem:psi_norm}
  Let $W_1$, $W_2$ be random variables, then
  \begin{equation*}
    \|W_1 W_2\|_{\psi_1}\le \|W_1\|_{\psi_2} \|W_2\|_{\psi_2}\,.
  \end{equation*}
\end{lemma}
\begin{proof}
\begin{align*}
  &\left\|\frac{W_1 W_2}{\|W_1\|_{\psi_2}\|W_2\|_{\psi_2}}\right\|_{\psi_1}
  \le \left\|\frac{W_1^2}{2\|W_1\|_{\psi_2}^2}+\frac{W_2^2}{2\|W_2\|_{\psi_2}^2}\right\|_{\psi_1}\\
  \le & \frac{1}{2}\left\|\frac{W_1^2}{\|W_1\|_{\psi_2}^2}\right\|_{\psi_1}+\frac{1}{2}\left\|\frac{W_2^2}{\|W_2\|_{\psi_2}^2}\right\|_{\psi_1}= 1\,. \qedhere
\end{align*}  
\end{proof}

\begin{lemma}\label{lem:concentration} There exist universal positive constants $c_1$, $c_2$
 such that if \eqref{eq:tail_assumption} holds with 
  $(A_n^2+1)\le  c_1\sqrt{n/\log N}$ then the following holds.
\begin{equation}\label{eq:concentration_mu}
\mathbb P\left[\max_j\left|\frac{\hat\mu_j-\mu_j}{\sigma_j}\right|\ge c_2 A_n\sqrt{\log N/n}\right]\le N^{-1}\,.  
\end{equation}
\begin{equation}\label{eq:concentration_sigma}
\mathbb P\left[\max_j\left|\frac{\hat\sigma^2_j}{\sigma^2_j}-1\right|\ge c_2 A_n^2 \sqrt{\log N/n}\right]\le N^{-1}\,.
\end{equation}
\begin{equation}\label{eq:concentration_gamma}
  \mathbb P\left[\max_{j,j'}\left|\hat \Gamma_{j,j'}-\Gamma_{j,j'}\right|\ge c_2 A_n^2\sqrt{\log N/n}\right]\le N^{-1}.
\end{equation}
\end{lemma}

\begin{proof}[Proof of \Cref{lem:concentration}]
  Since a smaller $A_n$ means better concentration, 
  without loss of generality we can assume that $A_n\ge 1$.
  
  The first inequality directly follows from the fact that $\|Z\|_{\psi_1}\le \|Z\|_{\psi_2}$,
  the Bernstein's inequality using Orlicz norm (Lemma 2.2.11 of \cite{vdvWellner}), and union bound.
  
  For the second inequality, note that
  \begin{align}
    &\left\|\left|\left(\frac{\xi_j-\mu_j}{\sigma_j}\right)^2-1\right|^{1/2}\right\|_{\psi_1}\le \left\|\frac{\xi_j-\mu_j}{\sigma_j}\right\|_{\psi_1}+\|1\|_{\psi_1}
    = A_n + (\log 2)^{-1}\,.\label{eq:psi_1_bound_sigma}
  \end{align}
  
We can re-write $\hat\sigma_j^2/\sigma_j^2-1$ as
  \begin{align*}
    \frac{\hat\sigma_j^2-\sigma_j^2}{\sigma_j^2}
    =n_{\rm te}^{-1}\sum_{i\in I_{\rm te}}
    \left[\left(\frac{\xi_j^{(i)}-\mu_j}{\sigma_j}\right)^2-1\right]-
    \left(\frac{\hat\mu_j-\mu_j}{\sigma_j}\right)^2\,.
  \end{align*}
Then the claim follows by applying Cramer's moderate deviation (see, e.g., Theorem 2.13 of \cite{PenaLS08})
with centered and scaled random variable $A_n^{-2}[(\xi_j-\mu_j)^2/\sigma_j^2 - 1]$ and deviation
$c\sqrt{\log N}$, together with \eqref{eq:concentration_mu} and union bound.

For the third inequality,
 we can re-write the off-diagonal error in $\Sigma$
  \begin{align*}
    \hat\Sigma_{j,j'}-\Sigma_{j,j'}=n_{\rm te}^{-1}\sum_{i\in I_{\rm te}} (\xi_j^{(i)}-\mu_j)(\xi_{j'}^{(i)}-\mu_{j'}) - (\hat\mu_j-\mu_j)(\hat\mu_{j'}-\mu_{j'})
  \end{align*}

A similar argument as \eqref{eq:psi_1_bound_sigma} shows that 
$$
\left\|\left|\frac{\xi_j-\mu_j}{\sigma_j}\frac{\xi_{j'}-\mu_{j'}}{\sigma_{j'}}-\Gamma_{j,j'}\right|^{1/2}\right\|_{\psi_1}\le
A_n+(\log 2)^{-1}\,.
$$
Then the same argument leads to the desired result by combining the deviation bound on
$\hat\Sigma-\Sigma$ and \eqref{eq:concentration_sigma}.
  %
%
\end{proof}

\begin{proof}[Proof of \Cref{thm:conditional}]
  Let $E$ denote the event that $\max_{j}|\hat\mu_j-\mu_j|/\sigma_j\le c_2 A_n\sqrt{\log N/n}$, $\max_j |\hat\sigma_j/\sigma_j-1|\le c_2 A_n^2 \sqrt{\log N/n}$, and $\max_{j,j'}|\hat\Gamma_{j,j'}-\Gamma_{j,j'}|\le c_2 A_n^2 \sqrt{\log N/n}$.
  \Cref{lem:concentration} implies that under the condition of theorem we have $\mathbb P(E)\ge 1-3 N^{-1}$.
  
  Let $\Lambda$ and $\hat\Lambda$ be the diagonal matrix corresponding to $\Sigma$ and $\hat\Sigma$,
  respectively. 
  For any positive semidefinite matrix $\Sigma$, let $Z_{\Sigma}$ be an $N(0,\Sigma)$ random vector.
  For $\alpha\in (0,1)$, let $z(\alpha,\Sigma)$
  be the upper $\alpha$ quantile of the maximum of $Z_{\Sigma}$.  
  
  For the first part of the theorem, by assumption we have $\max_{j}\mu_j/\sigma_j\le \gamma_n$
  with $\gamma_n=r_n\sqrt{1/(n\log N)}$. Recall that $\hat p$ is the bootstrap $p$-value given in \Cref{sec:p_val}.
  Here we ignore the bootstrap sampling variability and regard $\hat p$ as the limiting
  value when the bootstrap sample $B\rightarrow\infty$. We have
\begin{align}
  &\mathbb P(\hat p\le \alpha)\nonumber\\
   = &\mathbb P \left[\sqrt{n_{\rm te}}\max_j \left(\frac{\hat\mu_j - \mu_j}{\hat\sigma_j} + \frac{\mu_j}{\hat\sigma_j} \right) \ge z(\alpha,\hat \Gamma)\right]\nonumber\\
   \le & \mathbb P \left[\sqrt{n_{\rm te}}\max_j \left(\frac{\hat\mu_j - \mu_j}{\hat\sigma_j} + \frac{\mu_j}{\hat\sigma_j} \right) \ge z(\alpha,\hat \Gamma),~E\right] +\mathbb P(E^c)\nonumber \\
   \le & \mathbb P \left[\sqrt{n_{\rm te}}\max_j \frac{\hat\mu_j - \mu_j}{\sigma_j} \ge z(\alpha,\hat \Gamma)-c\sqrt{n}\gamma_n-c_2^2A_n^3\frac{\log N}{\sqrt{n}},~E\right] +\mathbb P(E^c) \label{eq:E_approx}\\ 
   \le & \mathbb P \left[\sqrt{n_{\rm te}}\max_j \frac{\hat\mu_j - \mu_j}{\sigma_j} \ge z(\alpha+\delta,\Gamma)-c \sqrt{n}\gamma_n-c_2^2A_n^3\frac{\log N}{\sqrt{n}},~E\right] +\mathbb P(E^c)\label{eq:delta_approx} \\  
  \le & \mathbb P\left[\max Z_{\Gamma} \ge z(\alpha+\delta,\Gamma)-c\sqrt{n} \gamma_n-c_2^2A_n^3\frac{\log N}{\sqrt{n}}\right] + C_1 n^{-C_2} + 3 N^{-1}\label{eq:Gauss_approx}\\
  \le &\alpha+\delta + \left(c \sqrt{n}\gamma_n+c_2^2A_n^3\frac{\log N}{\sqrt{n}}\right)2(1+\sqrt{2\log N}) +C_1 n^{-C_2} + 3 N^{-1}\label{eq:anti_concentration}\\
  \le &\alpha+c\max\left(\left(\frac{A_n^4\log^5 N}{n}\right)^{\frac{1}{6}},
  r_n,\left(\frac{A_n^6\log^3 N}{n}\right)^{\frac{1}{2}},N^{-(C_2\wedge 1)}\right)\,,
\end{align}
where
\begin{itemize}
  \item \eqref{eq:E_approx} follows from the concentration of $\hat\mu_j-\mu_j$ and $\hat\sigma_j/\sigma_j-1$, and the definition of $E$.
  \item \eqref{eq:delta_approx} follows from the difference between maxima of Gaussian random vectors (Theorem 2 of \cite{Cherno2}). Here $\delta = c A_n^3 (\log^{5/6}N) n^{-1/6}$ for some universal constant $c$.
  \item \eqref{eq:Gauss_approx} follows from the Gaussian comparison of maxima of mean vectors (Corollary 2.1 of \cite{Cherno1}) with the $\psi_1$  norm condition and
  $\log ^7N \le c n^{1-c'}$ for some constants $c,c'$. Here $C_1$ and $C_2$ in \eqref{eq:Gauss_approx} are universal constants involved in the Gaussian comparison result.
  \item \eqref{eq:anti_concentration} follows from the anti-concentration of maxima of Gaussian random vectors (Theorem 3 of \cite{Cherno2}).
\end{itemize}
When $\mu_j=0$ for all $j$, the above argument also goes through in the other direction. In this case we have $\mathbb P(\hat p\le \alpha)=\alpha+o(1)$.  The first part of the theorem is complete.

For the second part. Let $j$ be such that $\frac{\mu_j}{\sigma_j}\ge r\sqrt{\frac{\log N}{n}}$. Then
\begin{align*}
  \mathbb P(\hat p\le \alpha) =  & \mathbb P\left[\sqrt{n_{\rm te}}\max_{j'}\frac{\hat\mu_{j'}}{\hat\sigma_{j'}}\ge z(\alpha,\hat\Gamma)\right]\\
  \ge & \mathbb P\left[\sqrt{n_{\rm te}}\frac{\hat\mu_{j}}{\hat\sigma_j}\ge z(\alpha,\hat\Gamma)\right]\\
\ge & \mathbb P\left[\sqrt{n_{\rm te}}\frac{\hat\mu_{j}}{\hat\sigma_j}\ge z(\alpha,\hat\Gamma),E\right]
\end{align*}
On event $E$, we have
\begin{align*} \frac{\hat\mu_j}{\hat\sigma_j}=
  &\left(\frac{\mu_j}{\sigma_j}+\frac{\hat\mu_j-\mu_j}{\sigma_j}\right) \left(\frac{\sigma_j}{\hat\sigma_j}-1+1\right)
\\
  \ge &
  \left(r-c_2A_n\right)\sqrt{\frac{\log N}{n}}\left(1-c_2A_n^2\sqrt{\frac{\log N}{n}}\right)
  \ge c \sqrt{\log N/n}\,,
\end{align*}
for some constant $c>0$ when $n$ is large enough. Here the last inequality follows because
by assumption we have $r\ge c'(A_n \vee 1)$ for some large enough constant $c'$.

On the other hand, using Mill's inequality and union bound we can directly verify that
$z(\alpha,\hat\Gamma)\le \sqrt{2(\log (M-1)-\log\alpha)} \le 2\sqrt{\log N}$ whenever $\alpha\ge n^{-1}$. Thus when $n$ is large enough we always have $\sqrt{n_{\rm te}}\hat\mu_{j}/\hat\sigma_j\ge z(\alpha,\hat\Gamma)$ on $E$. As a result,
\begin{align*}
&\mathbb P(\hat p\le \alpha)\ge \mathbb P(E)\ge 1-3N^{-1}.  \qedhere  
\end{align*}
\end{proof}

\subsection{Proofs for \Cref{sec:linear}}\label{sec:proof_2}

\begin{proof}[Proof of \Cref{thm:linear}]
  Let $\mathcal M=\{1,2,...,M\}$.
  For any $m\in\mathcal M$, let $J_m$ be the subset of active variables in model $m$. Define
   $\mathcal M_1=\{m\in\mathcal M:J_{m^*}\backslash J_m\neq \emptyset\}$ be the candidate models missing at least one true variable,
   and $\mathcal M_2=\mathcal M\backslash \left(\mathcal M_1 \cup\{m^*\}\right)$.
  We use ``$a\lesssim b$'' to denote that
  $a\le c b$ for some universal constant $c$.

By consistency  of least square estimates we have $\sup_{m\in\mathcal M}\|\hat\beta_m-\beta_m\|_2=o_{P_{\rm tr}}(1)$, where $\beta_{m}= \Sigma_{J_m,J_m}^{-1}\mathbb E(X_{J_m}^T Y)$ is the population version, and the notation $o_{P_{\rm tr}}(\cdot)$
emphasizes that the randomness is on the training data $D_{\rm tr}$.
In the following we overload the notation of $\hat\beta_m$ and $\beta_m$ by embedding
them in $\mathbb R^p$, with zero values in the entries in $J_m^c$.

For two candidate models $m,m'$,
the difference of squared error is
  \begin{align*}
    \xi_{m,m'}=&(Y-X^T\hat\beta_m)^2-(Y-X^T\hat\beta_{m'})^2\\=&2\epsilon X^T(\hat\beta_{m'}-\hat\beta_m)+(X^T(\hat\beta_m-\beta))^2-(X^T(\hat\beta_{m'}-\beta))^2\,.
  \end{align*}

Now we show that for any $m\in \mathcal M_1$, we have $P(m\in\mathcal A_n)\rightarrow 0$.

Let $m'=m^*$, then we have
\begin{align}
  \mathbb E(\xi_{m,m^*}\mid D_{\rm tr})=&(\hat\beta_m-\beta)^T\Sigma(\hat\beta_m-\beta)+
  (\hat\beta_{m^*}-\beta)^T\Sigma(\hat\beta_{m^*}-\beta)\nonumber\\\ge& \beta_{\min}^2\lambda_{\min}+o_{P_{\rm tr}}(1)\label{eq:consistency_lowerbound_M1}
\end{align} %
%
so that
\begin{equation}\label{eq:consistency_mu_low}
  \hat\mu_{m,m^*}\ge \beta_{\min}^2\lambda_{\min} +o_P(1)\,.
\end{equation}
Since CVC uses a studentized test statistic, we still need to provide an upper bound for $\hat\sigma_{m,m^*}$. First we have, using the assumption that $X$ has finite fourth moment and $\epsilon X$ has finite second moment,
\begin{align}
  &{\rm Var}(\xi_{m,m^*}\mid D_{\rm tr})\nonumber\\=&4\sigma^2 (\hat\beta_m-\hat\beta_{m^*})^T\Sigma (\hat\beta_m-\hat\beta_{m^*})+\mathbb E\left[(\hat\beta_m+\hat\beta_{m^*}-2\beta)^T(XX^T-\Sigma)
  (\hat\beta_m-\hat\beta_{m^*})\right]^2\nonumber\\
  =& O_{P_{\rm tr}}\left(  \sigma^2 \lambda_{\max}\beta_{\max}^2+\beta_{\max}^4\lambda_{\max}^2\right)=O_{P_{\rm tr}}(1)\,.\label{eq:ols_xi_var_upper}
\end{align}
Thus
\begin{equation}\label{eq:consistency_var_up}
  \hat\sigma_{m,m^*}=O_{P}(1)\,.
\end{equation}
Combining \eqref{eq:consistency_mu_low} and \eqref{eq:consistency_var_up} we have
\begin{equation}\label{eq:consistency_student_low}
\mathbb P\left[\sqrt{n_{\rm te}}\frac{\hat\mu_{m,m^*}}{\hat\sigma_{m,m^*}}\ge \sqrt{2\log n}\right]\rightarrow 1\,.
\end{equation}
As a consequence, the rejection probability of model $m$ approaches $1$ as $n$ increases.

Now consider $m\in\mathcal M_2$. For brevity we denote $J_m$ as $J$ and 
  $J_{m^*}$ as $J^*$.
 Let $\vec\epsilon$ denote the realized vector of $\epsilon$ in $D_{\rm tr}$,
 $X_{J}$ denote the \emph{training data} columns of $X$ in $J$, 
 and $\Sigma_{J}$ denote the corresponding population covariance matrix.
  \begin{align*}
    &\mathbb E\left(\left[X^T (\hat\beta_m-\beta)\right]^2|D_{\rm tr}\right)\nonumber\\
     =& \vec\epsilon ^T X_{J} (X_{J}^T X_{J})^{-1} \Sigma_{J}(X_{J,\rm tr}^T X_{J})^{-1} X_{J}^T \vec\epsilon\\
    =&\vec\epsilon ^T X_{J} (X_{J}^T X_{J})^{-1}\left[ \Sigma_{J}-(X_{J}^T X_{J}/n_{\rm tr})\right](X_{J}^T X_{J})^{-1} X_{J}^T \vec\epsilon\\
    &\quad +
    n_{\rm tr}^{-1}\vec\epsilon ^T X_{J} (X_{J}^T X_{J})^{-1} (X_{J}^T X_{J})(X_{J}^T X_{J})^{-1} X_{J}^T \vec\epsilon\\
    =&\vec\epsilon ^T X_{J} (X_{J}^T X_{J})^{-1}\left[ \Sigma_{J}-(X_{J}^T X_{J}/n_{\rm tr})\right](X_{J}^T X_{J})^{-1} X_{J}^T \vec\epsilon\\
    &\quad +
    n_{\rm tr}^{-1}\vec\epsilon ^T X_{J} (X_{J}^T X_{J})^{-1}  X_{J}^T \vec\epsilon
  \end{align*}
  where $P_J$ is the projection matrix of the linear subspace spanned by the columns of $X_J$. 
  
  Therefore
  \begin{align*}
n_{\rm tr}  \mathbb E(\xi_{m,m^*}\mid D_{\rm tr}) &=\vec\epsilon ^T (P_{J}-P_{J^*}) \vec\epsilon + R_n\,.
  \end{align*}
  where
  $$
  R_n=n_{\rm tr}\vec\epsilon ^T X_{J} (X_{J}^T X_{J})^{-1}\left[ \Sigma_{J}-(X_{J}^T X_{J}/n_{\rm tr})\right](X_{J}^T X_{J})^{-1} X_{J}^T\,.
  $$

Since $P_J - P_{J^*}$ is a projection matrix, we can write $P_J-P_{J^*}=\sum_{l=1}^L v_l v_l^T$, where
$L\ge 1$ is the number of irrelevant variables in model $J$, and $v_1,...,v_L$ are orthonormal vectors in $\mathbb R^{n_{\rm tr}}$. Thus
$\vec\epsilon ^T (P_{J}-P_{J^*}) \vec\epsilon\ge (v_1^T\vec\epsilon)^2$.
Note that $\|v_1\|_\infty \le \max_{1\le i\le n} H_{ii}^{1/2}$ where $H$ is the hat matrix with all columns of data.  Let $r_n$ be a positive sequence such that $r_n\uparrow\infty$ and $r_n =o(n^{1/4})$, then we have
\begin{align*}
  &\mathbb P\left[\mathbb E\left(\xi_{m,m^*}\mid D_{\rm tr}\right)<0\right]\le
  \mathbb P\left[(v_1^T\vec\epsilon)^2+R_n<0\right]\\
  \le &\mathbb P\left[|v_1^T\vec\epsilon|< |R_n|^{1/2}\right]\\
  \le & \mathbb P\left[|v_1^T\vec\epsilon|< r_n^{1/2}n^{-1/4}\right]+
  \mathbb P\left[|R_n|> r_n n^{-1/2}\right]\\
  \le & \mathbb P\left[|v_1^T\vec\epsilon|< r_n^{1/2}n^{-1/4}~\big|~ \|v_1\|_\infty \le 2 r_n^{1/2}n^{-1/4}\right] \\
  &\quad+ \mathbb P\left[\|v_1\|_\infty > 2r_n^{1/2}n^{-1/4}\right]+\mathbb P\left[|R_n|> r_n n^{-1/2}\right]\\
  =& I + II + III\,.
\end{align*}
For term $I$, we use anti-concentration inequality. 
Because $\epsilon$ has mean zero and positive variance, there exist constants $L>0$ and $\eta\in(0,1)$, such that
$\sup_{t}P(\epsilon\in[t,t+L])\le 1-\eta$.  
Let $v_1=(v_{1,j}:1\le j\le n)$.  If $v_{1,j}\neq 0$, then $\sup_{t}P(v_{1,j}\epsilon_j\in[t,t+|v_{1,j}|L])\le 1-\eta$.
Now applying the anti-concentration inequality for sums of independent random variables (Theorem 1 of \cite{Rogozin61}), we have, when $\|v_1\|_\infty\le 2r_n^{1/2}n^{-1/4}$,
\begin{align*}
  &\mathbb P\left(|v_1^T\vec\epsilon|\le r_n^{1/2}n^{-1/4}\big| v_1\right)\\
  \lesssim &\frac{r_n^{1/2}n^{-1/4}}{\sqrt{\sum_{j:v_{1,j}\neq 0} v_{1,j}^2L^2(1-\eta)}}
  \lesssim r_n^{1/2}n^{-1/4}=o(1)\,.
\end{align*}

For term $II$, we have
\begin{align*}
  \mathbb P\left[\|v_1\|_\infty > 2r_n^{1/2}n^{-1/4}\right]
  \le  \mathbb P\left[\max_i H_{ii}> 4 r_n n^{-1/2}\right]=o(1)
\end{align*}
where the last equation follows from the assumption that $\max_i H_{ii}=o_P(n^{-1/2})$.

For term $III$, we have $III=o(1)$ since $r_n\rightarrow\infty$ and $R_n=O_P(n^{-1/2})$, according to the assumption that $\epsilon X$ has finite second moment and $X$ has finite fourth moment.

Now we have shown that
\begin{equation}
  \mathbb P\left[\mathbb E(\xi_{m,m^*}\mid D_{\rm tr})<0\right]=o(1)\,,~~\forall~ m\in\mathcal M_2\,.\label{eq:consistency_lowerbound_M2}
\end{equation}
Combining \eqref{eq:consistency_lowerbound_M1} and \eqref{eq:consistency_lowerbound_M2} we have
$$
\mathbb P\left[\sup_{m\neq m^*}\mathbb E(\xi_{m^*,m}\mid D_{\rm tr})>0\right]=o(1)\,.
$$

Now conditioning on the event $\sup_{m\neq m^*}\mathbb E(\xi_{m^*,m} \mid D_{\rm tr})<0$, we have
\begin{align}
  &\mathbb P\left[\sup_{m\neq m^*}\sqrt{n_{\rm te}}\frac{\hat\mu_{m^*,m}}{\hat\sigma_{m^*,m}}\ge z(\alpha_n,\hat\Gamma)\right]\label{eq:consistency_cover_true}\\
  \le & \mathbb P\left[\sup_{m\neq m^*}\sqrt{n_{\rm te}}\frac{\hat\mu_{m^*,m}-\mu_{m^*,m}}{\hat\sigma_{m^*,m}}\ge z(\alpha_n,\hat\Gamma)\right]\nonumber\\
  =&o(1)\,,\nonumber
\end{align}
where the last step we use the fact that $\sqrt{n_{\rm te}}\frac{\hat\mu_{m^*,m}-\mu_{m^*,m}}{\hat\sigma_{m^*,m}}=O_{P_{\rm te}}(1)$ according to the self-normalized moderate deviation result (Theorem 7.4 of \cite{PenaLS08}), which holds since by condition B2 $\xi_{m^*,m}$ has finite third moment provided that $\hat \beta_{J_m^*}$ and $\hat \beta_{J_m}$ are bounded (with probability $1-o_{P_{\rm tr}}(1)$).
\end{proof}

\begin{proof}[Proof of \Cref{thm:linear_vfold}]
For given $v\in \{1,2,...,V\}$, we treat $I_v$ as $I_{\rm te}$ and $I_{-v}$ as $I_{\rm tr}$.  Using the same reasoning as from \eqref{eq:consistency_lowerbound_M1} to \eqref{eq:consistency_var_up} on the $v$th fold and then using union bound over $v=1,...,V$, we have, if $m\in\mathcal M_1$, then with probability $1-o(1)$
$$
\inf_{v=1,...,V}\frac{\hat \mu_{m,m^*}^{(v)}}{\hat\sigma_{m,m^*}^{(v)}}\ge c_3
$$
where $c_3>0$ is a constant depending on the distribution of $(X,\epsilon)$.

As a result,
\begin{align*}
  \frac{\hat\mu_{m,m^*}}{\hat\sigma_{m,m^*}}=
  \frac{\frac{1}{V}\sum_{v=1}^V \hat\mu_{m,m^*}^{(v)}}{\sqrt{\frac{1}{V}\sum_{v=1}^V(\hat\sigma_{m,m^*}^{(v)})^2}}\ge \frac{\frac{1}{V}c_3\sum_{v=1}^V\hat\sigma_{m,m^*}^{(v)}}{\sqrt{\frac{1}{V}\sum_{v=1}^V(\hat\sigma_{m,m^*}^{(v)})^2}}\ge \frac{c_3}{\sqrt{V}}\,.
\end{align*}
When $n$ is large so that $c_3\sqrt{n/V}\ge \sqrt{2\log N}$, then (let $\hat\Gamma$ be the sample correlation matrix of $[\xi_{m,j}^{(i)}-\hat\mu_{m,j}^{(v_i)}:j\neq m]_{i=1,...,n}$)
\begin{align*}
  \mathbb P(\hat p_{{\rm cv},m}\le \alpha)\ge & \mathbb P\left[\sqrt{n}\frac{\hat\mu_{m,m^*}}{\hat\sigma_{m,m^*}}\ge z(\alpha_n,\hat\Gamma)\right]
  \ge  P\left[\sqrt{n}\frac{\hat\mu_{m,m^*}}{\hat\sigma_{m,m^*}}
  \ge \sqrt{2\log N}\right]\\
  \ge & P\left[\sqrt{n}\frac{\hat\mu_{m,m^*}}{\hat\sigma_{m,m^*}}\ge c_3\sqrt{n/V}\right]=1-o(1)\,.
\end{align*}

Next we prove $\mathbb P(m^*\in\mathcal A_{\rm cv})\rightarrow 1$.
Applying the same argument as in the proof of \Cref{thm:linear} to each fold
 we have $$\mathbb P(\mu_{m^*,m}^{(v)}<0,~\forall~m\in\mathcal M\backslash\{m^*\})=1-o(1)\,.$$
Conditioning on this event, and using the same reasoning as in \eqref{eq:consistency_cover_true} we have
\begin{align*}
&\mathbb P\left[\sqrt{n}\frac{\hat\mu_{m^*,m}^{(v)}}{\hat\sigma_{m^*,m}^{(v)}}\ge z(\alpha_n,\hat\Gamma)\right]\\
\le & P\left[\sqrt{n}\frac{\hat\mu_{m^*,m}^{(v)}-\mu_{m^*,m}^{(v)}}{\hat\sigma_{m^*,m}^{(v)}}\ge z(\alpha_n,\hat\Gamma)\right]\\
=&o(1)\,,
\end{align*}
for all $m\neq m^*$ and $1\le v\le V$.
Finally using union bound we have
\begin{align*}
  \mathbb P(m^*\notin \mathcal A_{\rm cv})=&\mathbb P\left[\hat p_{{\rm cv}, m^*}\ge z(\alpha_n,\hat\Gamma)\right]\\
  = & \mathbb P\left[\sup_{m\neq m^*}\sqrt{n}\frac{\hat\mu_{m^*,m}}{\hat\sigma_{m^*,m}}\ge z(\alpha_n,\hat\Gamma)\right]\\
  \le &\mathbb P\left[\sup_{m\neq m^*}\sup_{1\le v\le V}\sqrt{n}\frac{\hat\mu_{m^*,m}^{(v)}}{\hat\sigma_{m^*,m}^{(v)}}\ge z(\alpha_n,\hat\Gamma)\right]\\
  = & o(1)\,.\qedhere
\end{align*}
\end{proof}

%


\bibliographystyle{plain}
\bibliography{cvc}
\end{document}